\documentclass[a4paper,11pt]{article}

\newif\ifdraft
\drafttrue

\usepackage{style}
\usepackage{shortcuts}
\usepackage{xcolor}

\algnewcommand\algorithmicforeach{\textbf{for each}}
\algdef{S}[FOR]{ForEach}[1]{\algorithmicforeach\ #1\ \algorithmicdo}

\newcommand*\samethanks[1][\value{footnote}]{\footnotemark[#1]}
\def\proofsubparagraph{\subparagraph}

\title{Distributed Santa Claus via Global Rounding\footnote{This research was funded in whole or in part by the Austrian Science Fund (FWF) \url{https://doi.org/10.55776/P36280} and \url{https://doi.org/10.55776/I6915}. For open access purposes, the author has applied a CC BY public copyright license to any author-accepted manuscript version arising from this submission.
This research was also supported by Dutch Research Council (NWO) project "The Twilight Zone of Efficiency: Optimality of Quasi-Polynomial Time Algorithms" [grant number OCEN.W.21.268].}}
\author{Tijn de Vos\thanks{TU Graz} \textcircled{r}\footnote{The author ordering was randomized using \url{https://www.aeaweb.org/journals/policies/random-author-order/} generator. It is requested that citations of this work list the authors separated by \texttt{\textbackslash textcircled\{r\}} instead of commas.}\hspace{.5em} Leo Wennmann\thanks{University of Southern Denmark} \textcircled{r}\hspace{.5em} Malte Baumecker\samethanks[2] \\ \textcircled{r}\hspace{.5em} Yannic Maus\samethanks[2] \textcircled{r}\hspace{.5em} Florian Schager\samethanks[2]}
\date{\today}

\begin{document}

\maketitle
\begin{abstract}
    In this paper, we initiate the study of a new class of problems in the \CONGEST model: Mixed packing \emph{and} covering linear programs (LP). 
    Previously, the design of optimization algorithms in the distributed setting was heavily focused on solving linear programs with either packing or covering constraints. We are the first to explore the class of linear programs with both packing \emph{and} covering constraints by providing a general-purpose \CONGEST solver for such LPs and by studying the sequentially well-studied Santa Claus problem as a central representative. 
    This NP-hard problem can be modeled as a bipartite graph of children and gifts where an edge indicates that a child desires a gift.  
    The goal is to assign the gifts to the children such that the least happy child is as happy as possible.  
    Even though this is a well-studied problem in the sequential setting, we provide the first results in the distributed setting. 
    In particular, we show that the complexity of computing an~$\Oo(\log n/\log \log n)$-approximation is $\widehat \Theta(\sqrt n+D)$ rounds, where our $\widetilde\Omega(\sqrt n+D)$-round lower bound even holds for \emph{any} approximation.  
    
    This lower bound is in stark contrast to related allocation tasks like maximum matching and load balancing, where there exist polylogarithmic-round approximation algorithms. 
    Since Santa Claus is a max-min allocation problem, there are many more \emph{global} dependencies. 
    Our lower bound construction shows that the non-existence of short augmenting paths does not guarantee a good approximation, thereby ruling out many techniques used for matching and related problems. Moreover, it shows that unlike packing \emph{or} covering LPs, there cannot be an exact LP solver for fractional mixed packing and covering LPs with a runtime independent of the diameter.
    We provide the first \CONGEST solver that computes a fractional solution to any mixed packing and covering LP that takes $\Oo(D \cdot \poly (\log n, \varepsilon^{-1}))$ rounds.

    To solve the Santa Claus problem, we model it as a mixed packing and covering LP.
     Based on a fractional solution obtained by our solver, we compute an integral solution via an involved global rounding approach. This rounding scheme might be of independent interest to the distributed community as its core routine is applicable to any bipartite graph.
    The aim of our paper is to extend the scope of efficiently solvable distributed LPs and to establish a novel fundamental separation in distributed optimization, thereby advancing our knowledge on the complexity of approximation in distributed optimization. 
\end{abstract}

\newpage
\tableofcontents

\newpage

\clearpage
\section{Introduction}
A wide variety of optimization problems have attracted considerable attention in the distributed setting~\cite{KMW06,SarmaHKKNPPW12,HarrisSS16,KMW16,AhmadiKO18,BachrachCDELP19,Balliu0HORS19,Ghaffari19,ABL20,BO20,EfronGK20,BeckerFKL21,RozhonGHZL22,ZuzicGYHS22,CL21,ChangL23,FGGKR_SODA23,HuangS23,IzumiKY24,MitrovicS25,ChangS26}. Their relevant complexities may be roughly $\Theta(\poly\log n)$, $\Theta(D), \Theta(\sqrt{n}+D)$, $\Theta(n)$, or even $\Theta(n^2)$ rounds  in the bandwidth-restricted \CONGEST model.\footnote{Throughout, $n$ denotes the number of nodes and $D$ denotes the diameter of the graph.} Understanding which optimization problems fall into which regime remains a central challenge. 
Linear programming provides a useful lens on this landscape. Many distributed optimization problems can be formulated as linear programs, and for some important classes we already have a fairly complete algorithmic picture. Fractional packing and covering LPs can be solved in polylogarithmic \CONGEST rounds \cite{KMW06}, and the corresponding integer programs admit polylogarithmic algorithms in the stronger \LOCAL model \cite{GKM17,ChangL23}. In contrast, the best general \CONGEST solver for linear programs is substantially slower, taking $\Ohat(n+D\sqrt n)$ rounds~\cite{Vos23}.\footnote{Throughout, we write $\Otilde(f):=\Oo(f\poly\log f)$ and $\Ohat(f) := \Oo(f^{1+o(1)})$ for any function $f$.}

In this paper, we investigate a natural class lying between packing \emph{or} covering LPs and general LPs: linear programs containing both packing and covering constraints.
As one ingredient of our approach, we provide a distributed solver to obtain a fractional solution. Our main focus, however, is on the additional difficulties that arise when seeking integral solutions, which we study through the \emph{Santa Claus problem}, a classical max-min allocation problem that can be formulated as a linear program that uses both packing and covering constraints. We obtain an almost-tight distributed complexity characterization for its approximation.  Our main technical contribution is a global rounding procedure that converts a fractional solution into an integral assignment. 

More broadly, we hope that this work helps bring mixed packing-covering problems into the scope of distributed optimization. Recent follow-up work on a preliminary version of this paper indicates continued interest in this direction. In \cite{BernsteinGW26}, the authors further develop distributed mixed packing-covering solvers and apply them to related load-balancing problems. See the discussion after \Cref{thm:simplified-mpc-solver} for more details. 

\medskip 

The Santa Claus problem, first coined in the seminal paper of Bansal and Sviridenko~\cite{BansalS06}, is defined as follows.
Informally,
we help Santa Claus in his annual quest of trying to fulfill the children's gift wishlists by distributing gifts to the children such that the least happy child is as happy as possible.
More formally, we are given a set of gifts $\cG$ and a set of children~$\cC$ where each gift $g \in \cG$ has a fixed value $v_g \ge 0$.
Encoding the children's wishlists as a bipartite graph $G = (\cC\cup \cG,E)$, our goal is to find disjoint sets of gifts $\cG_c \subseteq N(c) \subseteq \cG$ for all children $c \in \cC$ according to the objective function
\begin{equation*}
    \max \min_{c \in \cC} \sum_{g \in \cG_c} v_g.
\end{equation*}
Santa Claus, also referred to as the restricted assignment case or max-min fair allocation, has been extensively studied in the sequential setting \cite{BezakovaD05,BansalS06, Feige08, AsadpourFS12, AnnamalaiKS15, PolacekS16, ChengM18, JansenR20, DaviesRZ20, ChengM22, HaxellS25}. 
Note that some papers consider the scheduling formulation of Santa Claus, where the goal is to find an assignment of jobs (gifts) to machines (children) that maximizes the minimum machine load. 
Other variants of Santa Claus like the matroid, submodular, budgeted, unrestricted, and online versions have also been studied \cite{ChakrabartyCK09,DaviesRZ20, BamasGR21,SpringerHPK22, BamasLMRS24, BamasMR25, RohwedderRW25}.

\subsection{Our Contribution}
Notably, Santa Claus has not yet been studied in a distributed setting.
In this paper, we consider the \CONGEST model~\cite{peleg00}, where a communication network is abstracted as an $n$-node graph, and the vertices communicate with each other through the edges of the graph in synchronous rounds by sending a message of $O(\log n)$ bits to each neighbor. The complexity measure is the number of communication rounds  until every node has produced its output, e.g., each child knows which gift it receives and each gift knows to which child it is assigned. 

\paragraph{Distributed Hardness.}
The max-min property of the Santa Claus problem creates global dependencies that require a considerable amount of communication. 
This is reflected by the fact that computing any multiplicative approximation---an integral solution that is only a multiplicative factor worse than the optimum---is already hard.
More formally, we show the following lower bound.

\begin{restatable}{theorem}{ThmSantaClausLowerBound}\label{thm:SantaClausLowerBound}
    Any \CONGEST algorithm that computes any multiplicative approximation of the integral Santa Claus problem requires $\widetilde \Omega(\sqrt n+ D)$ rounds in the worst case.
\end{restatable}

This lower bound implies a novel separation in distributed complexity. It shows that any integer mixed packing and covering solver requires $\widetilde\Omega(\sqrt n+D)$ rounds. Moreover, with similar techniques, we show that the Santa Claus problem, and hence any integer mixed packing and covering solver, requires $\widetilde \Omega(D)$, even in the stronger \LOCAL model (see \Cref{lem:integral_local_diameter_LB} for the formal statement). This is in contrast to integer packing \emph{or} covering LPs, which can be solved in polylogarithmic rounds in the \LOCAL model~\cite{ChangL23}. 

In particular, \Cref{thm:SantaClausLowerBound} implies that even computing a \emph{polynomial} approximation is hard. 
This lower bound is in stark contrast to other allocation problems like matching and load balancing. 
In particular, bipartite perfect matching is a special case of Santa Claus: Consider the case where the number of children is equal to the number of gifts, then each child must receive exactly one gift to maximize the minimum happiness. 
Consequently, such an assignment exists if and only if there is a perfect matching. 
For maximum matching, there is a structural theorem by~\cite{HopcroftKarp} with the following property: If there are no short augmenting paths in a given feasible matching, then it is a good approximation of the optimum. 
The proof of \Cref{thm:SantaClausLowerBound} shows there is \emph{no} such property for Santa Claus: there exist instances with arbitrarily bad approximate solutions that do not contain short augmenting paths (see \Cref{rem:no_augmenting_path,rem:no_augmenting_path2}).
Therefore, many of the approaches and techniques used to solve matching simply do not transfer. 
The same argument also juxtaposes Santa Claus and its dual---the min-max allocation problem called \emph{load balancing}. 
As for matching, there are (polylogarithmic) algorithms for load balancing in both \LOCAL and \CONGEST~\cite{FHS15, HKPR18, OBL18, BO20, ABL20,BernsteinGW26}, which often rely on augmenting  paths. Hence, the techniques used in solving the load balancing problem do not transfer to Santa Claus either.

\paragraph{\boldmath\CONGEST Algorithm for Santa Claus.}

As our main contribution, we present the first multiplicative approximation for Santa Claus that can be computed in an almost-optimal number of \CONGEST rounds.

\begin{restatable}{theorem}{SantaApproximation}
    \label{thm:santa-approximation}
    There exists a randomized \CONGEST algorithm that, given a weighted bipartite graph $G=(\cC\cup \cG,E)$ on $n$ vertices and vertex values $v_g\in \R$ with $1/\poly n\le v_g\le \poly n$ for all~$g\in \cG$, computes \whp an \SantaApproxFactor-approximation for the Santa Claus problem in $\Ohat(\sqrt n+D)$ rounds.
\end{restatable}

At a high level, we start by computing a fractional solution to a suitable LP relaxation of the problem, and then employ a combination of global and local rounding steps to compute an integral solution where each child is assigned a subset of gifts such that the total value is not much worse than the optimum. 
In particular, we introduce a global rounding scheme which might be of independent interest. 
Its core part is an efficient subroutine that solves the following problem on any bipartite graph: given fractional edge weights, it finds a new assignment of weights such that (1) each node retains the same total value of incident edge weights, and (2) the edges with non-integer values form a forest. Essentially, it removes all cycles from the fractional support. Such a method has been used for flow rounding~\cite{ForsterGLPSY21,RozhonGHZL22}; we give an efficient algorithm for this more general setting.

\paragraph{LPs with Packing and Covering Constraints.}
As we explain later, there is a suitable relaxation of Santa Claus to a mixed packing and covering  LP. This is advantageous if we can solve the respective LP formulation \emph{efficiently} in the distributed setting -- as we will show is the case.
As mentioned above, there are already fast solvers for either packing or covering LPs \cite{KMW06,GKM17,ChangL23} and a slower solver for more general LPs~\cite{Vos23}. 
Mixed packing and covering LPs, however, have not been studied thus far in the distributed setting. 
Building on prior sequential and parallel work~\cite{MRWZ16}, we complement the world of existing distributed LP solvers by providing the first LP solver for mixed packing and covering LPs (also known as positive LPs) in the \CONGEST model. 
Since we did not optimize the running time beyond the Santa Claus requirements, we consider the following LP solver a minor contribution in our paper.

\begin{theorem}[Simplified version of \Cref{thm:solving_mixed_packing_covering}]
    \label{thm:simplified-mpc-solver}
    There is a deterministic \CONGEST algorithm that provides a $(1+\eps)$-approximate solution for mixed packing and covering LPs in $\Otilde\left(\frac{D\log^3 n}{\eps^3}\right)$ rounds. 
\end{theorem}

A preliminary version of this paper already stimulated further research on mixed packing-covering LPs~\cite{BernsteinGW26}. 
 Their framework achieves a mixed packing-covering LP solver without any diameter dependency that yields improvements for load balancing.
However, our lower bound shows that this improvement \emph{cannot} carry over to Santa Claus.

\subsection{Technique in a Nutshell}\label{sec:technique}
In this section, we explain how to obtain our main result, an \SantaApproxFactor-approximation for distributed Santa Claus in the \CONGEST model.
On a high level, we follow the sequential approach of Bansal and Sviridenko~\cite{BansalS06} that consists of computing a fractional solution using an LP and then diligently rounding it to an integral solution. 
However, significant changes are necessary to make the sequential algorithm suitable for the distributed setting, which we point out in the following explanation.

Throughout, we consider the Santa Claus problem as the bipartite graph~$G = (\cC \cup \cG, E)$ where an edge~$\set{c,g} \in E$ is present if a child~$c \in \cC$ desires the gift~$g \in \cG$.
Further, let~$n = \abs{\cC \cup \cG}$, $m = \abs{E}$ and~$\alpha = \SantaApproxFactor$.

\subsubsection{From Exact Solution to Approximation}
By a standard binary search framework, the decision variant of Santa Claus is equivalent to an $\alpha$-approximation for the Santa Claus problem for some~$\alpha > 1$:
Given a threshold~$T \ge 0$, either compute a solution of value~$T/\alpha$ or determine that the optimal value of the Santa Claus problem is less than~$T$.\footnote{In particular, for any~$T$ larger than the optimal value the LP is no longer feasible, \ie there does not exist a solution that satisfies all constraints.}
Intuitively, the threshold~$T$ can be thought of as the optimal integral value found by binary search.
For the rest of this exposition, consider the parameter~$T$ as the globally known optimal (integral) value of the Santa Claus problem. In the actual algorithm, we determine $T$ via a global binary search without affecting the total runtime.

\subsubsection{Computing an Approximate LP Solution}
\label{subsubsec:compute-approximate-solution}

The first important step is to choose the right LP formulation for the decision variant of the Santa Claus problem.
The challenge is that sequential LPs do not correspond to communication graphs, and can often be solved exactly in polynomial time\footnote{In (sequential) approximation algorithms, the polynomial running time is often not even explicitly stated.} using a standard LP solver.
When using distributed LP techniques, we are restricted to LP formulations that correspond to the communication graph \emph{and} we can only compute approximate fractional solutions.
Here, we are heavily deviating from~\cite{BansalS06}, because they use the \emph{configuration LP}\footnote{In the configuration LP, there exist exponentially many variables for all possible subsets of gifts that have a total value of at least~$T$. In the sequential setting, this can be solved in polynomial time by finding a solution to the dual LP via a polynomial-time separation oracle and the Ellipsoid method.} that does not correspond to the communication graph between children and gifts and has exponentially many variables.

Therefore, we use the highly non-trivial, polynomial-size LP formulation due to Davies, Rothvoss and Zhang \cite{DaviesRZ20} that has a constant integrality gap, i.e., any optimal integer solution is a constant factor away from an optimal fractional solution.
In \cref{subsubsec:related-work-sequential-setting}, we explain why this does not directly lead to a constant factor approximation.
For this formulation, we partition the set of gifts~$\cG$ into \emph{big gifts}~$\cB := \{ g \in \cG : v_g \geq T/\alpha\}$ and \emph{small gifts}~$\cS := \{ g \in \cG : v_g < T/\alpha\}$. 
Importantly, as we are showing an $\alpha$-approximation, any child that receives one big gift is satisfied and does \emph{not} require any more gifts.
Let~$\cC_g \subseteq \cC$ denote the subset of children that are interested in the gift~$g \in \cG$. We refer to the following LP as \LP throughout.
\begin{align*}
    \sum\limits_{s \in \cS: c \in \cC_s} v_s u_{cs} &\geq T \color{gray}{{}\cdot \Big( 1-\sum_{b \in \cB: c \in \cC_b} u_{cb}\Big)}  && \forall c \in \cC \\
    \color{gray}{u_{cs}} &\color{gray}{\leq 1 -\sum_{b \in \cB: c \in \cC_{b}} u_{cb}} &&  \color{gray}{\forall s \in \cS, \forall c \in \cC_s} \\
    \sum\limits_{c \in \cC_g} u_{cg} &\leq 1 &&  \forall g \in \cG \\
    u_{cg} &\ge 0   && \forall c \in \cC, \forall g \in \cG
\end{align*}
Since this LP can be somewhat tricky to understand, it can be beneficial to ignore the gray parts and get some intuition for a simpler formulation\footnote{To be precise, there are two polynomial-size LP formulations without the gift size distinction (besides the configuration LP) whose integrality gaps can be unbounded or as large as~$\Omega(m)$, see~\cite{BansalS06} for more details.} without the distinction between gift sizes. 
Both formulations ensure that each child receives at least a total value of~$T$ and that all gifts are (fractionally) assigned at most once.
Additionally, the non-trivial formulation ensures that each child who is not satisfied with big gifts alone receives the rest of the value from more than~$\alpha$ small gifts---thereby ensuring that an optimal integral solution is only a constant factor away.
For an in-depth explanation of the constraints, see \cref{subsec:santa-lp-formulation}.

Using \Cref{thm:simplified-mpc-solver}, we can compute an approximate fractional solution to \LP\ in $\widetilde{\mathcal{O}}(D \cdot \poly(\log n, \varepsilon^{-1}))$ rounds of \CONGEST.
More precisely, we compute a feasible fractional solution~$u$ to $\mathrm{SantaLP}(T/2)$.
Based on the size distinction between gifts, let~$u = (x,y)$ where~$x$ and~$y$ denote the assignments for the big and small gifts, respectively.
From here on out, let~$G$ denote the bipartite gift graph based on the fractional solution~$u$ where an edge~$\set{c,g}$ of weight~$u_{cg} \in [0,1]$ is present if~$u_{cg} > 0$.
Similarly, let~$\Gx{x} = (\cC \cup \cB, E)$ denote the \emph{big gift graph} between children and big gifts where an edge~$\set{c,b} \in E$ of weight~$x_{cb}$ is present if~$x_{cb} > 0$.

\subsubsection{Efficient Sparsification via Global Rounding}
\label{subsubsec:eliminate-cycles}

A crucial recurring subroutine in the computation of our integral gift assignment in \cref{subsubsec:gift-assignment}, \ie our $\alpha$-approximation, is the efficient elimination of cycles in subgraphs of~$G$ via global rounding.
Since we aim to round the weights to integer values, we restrict the gift graph~$G$ to its fractional edges during the cycle rounding algorithm.
Our goal is to reduce~$G$ to a forest.
In the sequential setting, the algorithm is simple:
For all even cycles~$C$ in the bipartite graph~$G$, partition~$C$ into two disjoint matchings~$C_1$ and~$C_2$ and decrease the edge weight in~$C_1$ at the same rate as the edge weights in~$C_2$ increase.
Eventually, at least one edge weight becomes integral and we remove it from the graph.
We repeat this process until there are no more cycles left in~$G$.

However, a naive implementation of this sequential procedure is prohibitively expensive. Since each cycle might be of length $\mathcal{O}(n)$ and we are only guaranteed to remove one edge in each iteration, this leads to a total runtime of $\mathcal{O}(n m)$ rounds---a time in which one can trivially solve any problem in \CONGEST by gathering the whole instance in a node.
Therefore, our \CONGEST algorithm is substantially more involved.
Based on a cycle cover algorithm~\cite{ParterY19,RozhonGHZL22}, we design a cycle rounding algorithm that finds and eliminates many disjoint short cycles in parallel and takes~$\Ohat(\sqrt{ n } + D)$ rounds.
The high-level idea of our approach to eliminating cycles is to repeat the following steps restricted to the graph with only fractional edges:
\begin{enumerate}[(1)]
    \item Compute a \emph{low diameter decomposition}, obtaining clusters with small diameter and at most $m/10$ inter-cluster edges. 
    \item Within each cluster, carefully extract a bridgeless subgraph by identifying all bridges.
    \item Compute \emph{short cycle covers}: every edge appears in at least one cycle of length $n^{o(1)}$.
    \item  Select a maximal set of disjoint cycles using an MIS algorithm on a virtual graph. 
    \item Round all cycles of this maximal set in parallel such that at least one edge in each cycle becomes integral.
\end{enumerate}
We show that each iteration makes a~$1/n^{o(1)}$-fraction of the remaining edges integral. We then repeat all steps on all remaining fractional edges. 
After $\log m \cdot n^{o(1)}$ iterations, the graph~$G$ becomes a forest.

A priori, we cannot guarantee any progress once the graph becomes very sparse, that is, when there are less than $n-1+n/10$ edges left. 
In this case, the number of inter-cluster edges from the low diameter decomposition can be as large as $n/10$ such that the remainder is already a tree and no further edges are rounded. 
In order to avoid this and make our algorithm efficient, we iteratively interleave the two following procedures:
Firstly, we remove all edges incident to vertices of degree one.
Secondly, we contract all long paths of degree-two vertices into a single edge.
Together, this results in a graph $H$ where the minimum degree is greater than two. 
In particular, this implies that~$(9/10 - \eps)|E(H)| \ge |V(H)|-1$, and we can still make progress.
Since $H$ is not a subgraph of the communication graph $G$ (but also contains contractions), each `round' on $H$ takes several rounds on $G$. 
We show that we can simulate each round on $H$ in $\Otilde(\sqrt n+D)$ rounds on $G$. 
Multiplying this overhead with the number of rounds per iteration and the number of iterations gives $\Ohat(\sqrt n+D)$.

\subsubsection{Computing an Integral Gift Assignment}
\label{subsubsec:gift-assignment}

Based on the feasible fractional solution~$(x,y)$ to $\mathrm{SantaLP}(T/2)$, we explain how to compute an integral gift assignment~$(X,Z)$ such that each child receives a value of~$T/\alpha$.
Our algorithm consists of the following three steps where each step takes~$\Ohat(\sqrt n+D)$ rounds in \CONGEST.

\paragraph{Clustering of Big Gifts.}
With the goal to compute multiple clusters of big gifts, we use our cycle rounding algorithm to compute a fractional solution~$(x',y)$ that corresponds to an acyclic subgraph~$\Gx{x'} \subseteq \Gx{x}$. 
However, the resulting forest $\cT'$ does not yet have all properties we require to compute an integral assignment for all gifts.
Thus, as a next step, we compute a fractional assignment~$(x^*,y)$ where~$G_{x^*} \subseteq \Gx{x'}$ is a sub-forest $\cT^*\subseteq \cT'$ of children and big gifts with the \emph{crucial} property that for each tree $\mathcal T \in \cT^\ast$ we can always find one child in~$\cT$ to be satisfied with small gifts while all remaining children receive one big gift.
Recall that any child is satisfied by one big gift.

\paragraph{Randomized Selection of Small Gifts.}
For each tree~$\cT \in \cT^*$, we randomly select \emph{exactly one} child to receive only small gifts and propagate the choice to all other nodes in $\cT$.
Importantly, the random selection happens independently and in parallel for all~$\cT \in \cT^*$ such that all respective decisions whether a small gift is assigned to a tree are in fact \emph{independent}.
Using a concentration and union bound, each small gift is \whp (fractionally) assigned to at most~$\beta =\mathcal{O}(\log n / \log \log n)$ trees and exactly one child within each tree.
After downscaling the solution by $\beta$, each gift is fractionally assigned to at most one child, and each selected child receives a value of $T/(2\beta)$. 
Notably, the random selection directly implies the integral big gift assignment~$X$, \ie we simply root all~$\cT^*$ at their respective chosen children and assign the big gifts to their respective child node using the crucial property.

Naively, rooting a tree takes $D$ rounds. Since we consider the big gift forest~$G_{x^*}$, each~$\cT^*$ can have diameter~$n$. 
By using broadcast primitives on a compressed virtual graph, we show how to root all~$\cT^*$ in~$G_{x^*}$ in parallel in $\Otilde(\sqrt n+D)$ rounds.

\paragraph{Integral Assignment of Gifts.} Lastly, we consider the subgraph of small gifts and randomly selected children with the goal to compute an integral assignment of small gifts. 
Using our cycle rounding algorithm as a subroutine again, we first eliminate all cycles in our small gift graph---this requires a careful adaptation of the edge weights such that the rounding takes both the fractional solution and the values of the small gifts into account.
In a second step, we round our fractional small gift assignment to an integral assignment~$Z$ such that each gift is assigned to at most one child and each child receives a value of~$T/\alpha$ where~$\alpha = 4\beta$. 
The high-level idea of the rounding is as follows:
For a fixed child~$c \in \cC$ that receives small gifts, let~$k_c$ denote the (fractional) number of small gifts it receives in the fractional solution.
Then we round it such that each child receives~$\floor{k_c}$ small gifts and thus loses at most one small gift.
In the worst case, this is the most valuable gift of value at most~$T/ \alpha$ and each child receives a value of at least~$T / (2\beta) - \max_{s \in \cS} v_s = T/\alpha$.

\subsubsection{Lower Bound}
\label{subsubsec:lower-bound}

We prove \Cref{thm:SantaClausLowerBound} via a reduction from set-disjointness in the two-party model to the Santa Claus problem in \CONGEST. In particular, we construct a family of lower bound graphs where deciding whether Santa Claus admits a solution of size $1$ or $0$ corresponds to the decision whether a certain set-disjointness instance is a ``yes''- or a ``no''-instance. Observe that any multiplicative approximation is also able to distinguish between a $1$- and a $0$-instance. Combining our reduction with the restricted bandwidth of our lower bound graphs and the long-standing lower bound of~$\Omega({n})$ for set-disjointness in the two-party communication model~\cite{DBLP:journals/tcs/Razborov92}, we obtain the claimed lower bound.

\subsection{Related Work and Other Approaches}

\subsubsection{Distributed LP Solvers}
All distributed solvers for linear programs give a $(1\pm\eps)$-approximate optimal solution. 
The most general distributed LP solver by De Vos~\cite{Vos23} takes $\Ohat(n+D\sqrt n)$ \CONGEST rounds and computes a fractional solution to any LP where the constraint matrix only has non-zero entries corresponding to a vertex or edge in the graph. Designed to solve the maximum flow problem, it computes an exact, integral solution for maximum flow in the same runtime. For the class of packing or covering LPs, Kuhn, Moscibroda and Wattenhofer~\cite{KMW06} provide a polylogarithmic-round \CONGEST algorithm that provides a fractional solution. In the stronger \LOCAL model,  Ghaffari, Kuhn, and Maus~\cite{GKM17} gave an \emph{integral} LP solver for the same class of problems.
Using a network decomposition by~\cite{RozhonG20,GhaffariGHIR23}, their algorithm can also be made deterministic. 
Later, Chang and Li~\cite{ChangL23} improved their polylogarithmic runtime for randomized algorithms to~$\Otilde(\log n/\eps)$. Notably, there are often faster algorithms for specific packing and covering problems, see, e.g.~\cite{JRS02,KMW16,BEG18,BHR18}. 
Moreover, Chang and Su~\cite{ChangS26} provide a framework that transforms \LOCAL algorithms for many combinatorial optimization problems---including many packing or covering LPs---into \CONGEST algorithms with small overhead. This framework holds for all graphs excluding a fixed minor.

\subsubsection{Sequential Santa Claus Approximations}
\label{subsubsec:related-work-sequential-setting}

As a result of several papers~\cite{BansalS06, Feige08, AnnamalaiKS15, ChengM18, DaviesRZ20}, the current state-of-the-art for sequential Santa Claus is a $(4+\eps)$-approximation by Cheng and Mao~\cite{ChengM22}.
It remains an interesting open question whether the approximation factor can be improved to match the hardness result by~\cite{LenstraST90, BezakovaD05} who showed that unless~$P=NP$ there cannot exist a polynomial-time~$(2 - \eps)$-approximation for Santa Claus.
This raises the important question of why we follow the first approach of~\cite{BansalS06} in our \CONGEST algorithm and did not obtain a constant approximation based on the works of~\cite{Feige08, AnnamalaiKS15, ChengM18, DaviesRZ20, ChengM22}.
The answer is two-fold.

Bansal and Sviridenko~\cite{BansalS06} also show how to extend the approach followed in this paper to obtain an $\Oo(\log \log n/ \log \log \log n)$ approximation via the sequential \emph{Lov\'asz Local Lemma (LLL)}.
Improving upon these ideas with iterative applications of LLL, Feige~\cite{Feige08} showed the first constant factor approximation.\footnote{At the time, it was a non-constructive result that was later shown to be constructive by a result of~\cite{HaeuplerSS11} that made LLL constructive for the several applications in Feige's result.} 
As the applications of distributed LLL either require bounds on the number of small gifts that are adjacent to a big gift cluster (and vice versa) or subsampling of configurations of the configuration LP, it is highly unclear if this approach can be successfully transferred to the distributed setting.

The sequential algorithms of~\cite{AnnamalaiKS15, ChengM18, DaviesRZ20, ChengM22} are inherently \emph{global} and combinatorial, \ie they do not require to solve an LP for an initial fractional solution.
On a high level, they rely on polynomially many iterations of finding augmenting paths in the gift graph that improve upon the current gift assignment. 
Intuitively, one can hope that all augmenting paths are short as long as we are far away from an optimal solution, because this would allow us to augment many paths in parallel. 
However, the proof of our $\widetilde \Omega(\sqrt n+ D)$-round lower bound (see \Cref{rem:no_augmenting_path}) shows that such an approach must fail: 
Even if we are far away from an optimal solution, the augmenting paths can still be of length~$D$.
This is also the reason why the constant integrality gap of \LP does not automatically imply a constant approximation, as our approximation factor can only be as good as the guarantee that we maintain during our rounding procedure.
It remains an open question whether a different approach that does not rely on the augmenting paths leads to a better approximation factor for our \CONGEST algorithm.

A more general version of Santa Claus, also called \emph{unrestricted} Santa Claus, considers the case where each child~$c \in \cC$ desires a gift~$g \in \cG$ with a different value $v_{cg} \ge 0$.
Known as a notoriously difficult open problem, the state-of-the-art is a~$\max\set{\poly(\log |\cG|), |\cG|^{\eps}}$-approximation by Chakrabarty, Chuzhoy, and Khanna \cite{ChakrabartyCK09}.
In recent years, Bamas and Rohwedder \cite{BamasR23} showed a~$\poly(\log\log |\cG|)$-approximation for what Bateni, Charikar, and Guruswami \cite{BateniCG09} identified as another central special case of Santa Claus---the problem of max-min degree arborescence as it captures much of its difficulty.
However, this still leaves a significant gap to the known hardness by~\cite{LenstraST90, BezakovaD05} that also holds for the unrestricted case.

\subsubsection{Distributed Rounding}
\label{subsubsec:related-work-distributed-setting}

On a high level, our algorithm computes a fractional solution to Santa Claus and then rounds it to an integral solution. 
Although distributed rounding is well-studied, the challenges we face are distinct from existing work. 
Mainly following one recipe, previous work on local rounding in distributed algorithms can be described as follows: Start with a suitable fractional solution---usually a uniform distribution over all possible output labels---of the problem. 
Next, define a linear objective function~$\Phi = \sum_{v \in V} \Phi_v$.
Then pick a set~$S$ of vertices such that changing the fractional labels of any vertex~$v \in S$ does not affect $\Phi_w$ for any $w \in S \setminus \{ v \} $.
Round the fractional labels in $S$ to integral labels and remove $S$ from the graph.
Repeat this procedure until the graph is empty.
This produces a partial solution for a constant fraction of vertices in the graph. 
Hence, recursing on the remaining subgraph for $\mathcal{O}(\log n)$ iterations solves the problem on the entire graph.
This general method is responsible for many recent breakthroughs in the \LOCAL model; see \cite{Fischer20, GK21, GhaffariGR21,FGGKR_SODA23,GG23,GG24}.
This approach relies on a few key assumptions that are violated for Santa Claus: First of all, it assumes that any valid partial solution can always be extended to a proper solution of the same objective value.
Secondly, it assumes that it is always possible to round a fractional solution without loss in objective value.
This is clearly impossible for Santa Claus, as illustrated by the following simple instance: 
There are two children and one gift of value $1$ desired by both children. 
The optimal fractional solution has value $1 / 2$ but any integral assignment has value $0$.

For global problems, rounding is less common. 
The only rounding algorithm in \CONGEST we are aware of is for flow and takes $\Otilde(\sqrt m+D)$ rounds~\cite{ForsterGLPSY21}. 
The important subroutine of finding an Eulerian tour has since been improved to $\Otilde(\sqrt n+D)$ rounds~\cite{RozhonGHZL22} which can be seen as the starting point of the cycle rounding detailed in \cref{subsec:cycle-elimination}. 
However, we cannot use their algorithm, as our input graph does not have the structure required by~\cite{ForsterGLPSY21,RozhonGHZL22}.

\subsection{Future Directions}
\label{subsec:future-directions}

In some sense, the cycle rounding subroutine of \cref{subsec:cycle-elimination} can be seen as a sparsification procedure: it significantly reduces the number of fractional edges while maintaining the fact that it is a feasible fractional solution. 
To some extent, this rounding technique is content-oblivious: it is rounding the values of a fractional solution of a linear program.
Further research is required to find out under which circumstances such a sparsification can be used in other optimization problems. 
For Santa Claus, it already raises the question whether this sparsification can be run directly on the entire fractional solution---thereby considerably simplifying our algorithm. 
It turns out that this is not possible, which requires some care to see. 
Such a sparsification step essentially removes any edges that attain weight zero through rounding. 
However, an optimal integral solution does not need to be a subset of the non-zero fractional edges (see \Cref{example:sparsification} in Appendix~\ref{sec:example}).
Importantly, when we apply this sparsification step as a subroutine in our algorithm, we \emph{diligently} guarantee that there exists a good integral solution that is a subset of the fractional edges. 
We leave as an open question which other optimization problems can be approached in a similar manner. 

\subsection{Outline}
The remainder of this paper is structured as follows. In \Cref{subsec:cycle-elimination}, we provide our cycle rounding algorithm that reduces the number of fractional edges in a bipartite graph to a forest; this result is independent of the Santa Claus problem.  
In \Cref{sec:Santa_Claus}, we provide our \CONGEST approximation algorithm for the Santa Claus problem (see \Cref{thm:santa-approximation}). Next, in \Cref{sec:lower-bounds}, we provide our lower bounds (see \Cref{thm:SantaClausLowerBound,lem:integral_local_diameter_LB}). Finally, in \Cref{sec:LP_solver}, we provide our fractional mixed packing and covering LP solver (see \Cref{thm:simplified-mpc-solver}).

\newpage
\section{Efficient Sparsification via Global Rounding}
\label{subsec:cycle-elimination}
Given an approximate fractional solution to the Santa Claus problem, we want to reduce the number of edges with non-integer weights in the graph. 
For that purpose, we provide a \CONGEST algorithm that iteratively rounds edge weights along cycles in a bipartite\footnote{Since the input graph $G$ to the Santa Claus problem (see \Cref{sec:Santa_Claus}) consists of children and gifts, it is naturally bipartite.} graph such that the weight of at least one edge becomes integral.
Thus, we overcome the challenge of efficiently rounding edge weights along many cycles in parallel.
Repeating this procedure until the graph induced by the set of edges of non-integral weights becomes acyclic---in other words, a forest---yields the following main result in this section.

\begin{restatable}[Cycle Rounding]{lemma}{CycleRounding}
    \label{lem:cycle-elimination}
    Let $G=(V,E)$ be a bipartite \CONGEST communication network with edge weights $w\colon E \to [0,1]$. There exists a deterministic $\widehat{\mathcal{O}}(\sqrt n+D)$-round algorithm that computes edge weights $w'\colon E \to [0,1]$ such that 
   \begin{enumerate}
       \item $\sum_{v\in N(u)} w'(uv)= \sum_{v\in N(u)} w(uv)$ for every $u\in V$; and
       \item There exists a forest $F\subseteq  E$ such that all edges $e\in  E\setminus F$ have integral values $w'(e) \in \{0, 1\}$.
   \end{enumerate}
\end{restatable}

The idea of the rounding procedure is simple: for any (even) cycle in $G$, we find the edge of minimum or maximum weight~$w$ (depending on which is closer to $0$ or $1$) and alternately increase or decrease all weights around the cycle by $w$.
Although correctness of such an algorithm is straightforward, it is tricky to efficiently implement it in the distributed setting. There exist \CONGEST implementations of this procedure for rounding flow~\cite{ForsterGLPSY21,RozhonGHZL22,ForsterV23}, culminating in an~$\Otilde(\sqrt n + D)$-round algorithm. 
However, these algorithms are tailored to graphs with specific properties. In particular, $(1)$ the sum of all edge weights for any vertex needs to be 1 and $(2)$ the graph should be Eulerian, i.e., each node has even degree. 
The latter property implies that it can be decomposed into a union of cycles. 
Overcoming the challenges of dropping these assumptions requires several new ideas. 
In our approach, we bypass $(1)$ completely. 
We overcome $(2)$ since our new subroutine only requires the graph to be \emph{bridgeless}. 
A \emph{bridge} is an edge whose removal makes the graph disconnected and thus does not lie on a simple cycle. Any bridges that occur can be found easily and will be part of the forest $F$ as defined in \Cref{lem:cycle-elimination}.

In \Cref{subsubsec:cycle-rounding-algorithm}, we give an overview of our algorithm, together with the subroutines we use from the literature. Then, in \Cref{subsubsec:proof_cycle_rounding}, we give the algorithm in detail and the full proof.

\subsection{Algorithm Overview}
\label{subsubsec:cycle-rounding-algorithm}

The starting point of our algorithm is the subgraph $G_0 \subseteq G$ induced by all edges $e$ such that~$w(e) \notin \{ 0, 1 \} $.
Then, in iteration $i$, we perform the following modifications to the graph~$G_i$:
First, we iteratively apply two operations inspired by Miller and Reif \cite{MillerR89} called \textit{Rake} and \textit{Compress} to $G_i$.
The \textit{Rake} operation helps to prune edges from $G_i$, which are not part of any cycle, while the \textit{Compress} operation makes it feasible to round edge weights on very long cycles.
Another important feature of the \textit{Compress} iteration is that it makes the graph more dense. 
In particular, we will show that the number of edges $m_i$ of the compressed graph $H_i$ is at least $3 / 2$ times the number of vertices $n_i$ of $H_i$.
This is crucial, because in the second step of each iteration, we compute a low-diameter decomposition of $H_i$, which permits up to~$m_i/10$ edges to have endpoints in two different clusters of the decomposition.
In particular, without this step we would not be able to reduce a graph $G$ with $(1 + o(1)) \cdot n$ edges to a forest, since it could happen that all cycles in $G$ get lost after removing the edges between the different clusters.
However, since the number of edges with two endpoints in the same cluster is at least $9/10 \cdot 3/2 \cdot n_i > 5/4 n_i$, we are still guaranteed to find enough cycles to make progress.
In order to find a large set of disjoint short cycles, we first apply a procedure to each component that makes the component bridgeless.
Finally, we compute a cycle cover on the pruned components in parallel and select a maximal independent set of those cycles to round in this iteration.
In the rest of this section, we explain each of these steps in more detail.
For an overview we refer to \cref{alg:cycle-rounding}.

\begin{algorithm}
    \caption{Cycle Rounding}
    \label{alg:cycle-rounding}
    \SetKwInOut{Input}{Input}
    \SetKwInOut{Output}{Output}
    
 \newcommand\mycommfont[1]{\textcolor{gray}{\itshape #1}}
    \SetCommentSty{mycommfont}
    \SetKwComment{tcp}{}{}
    \DontPrintSemicolon

        \Input{A \CONGEST network $G=(V,E)$ with edge weights $w: E \to [0,1]$.}
        \Output{Edge weights $w': E \to [0,1]$ satisfying the properties in \Cref{lem:cycle-elimination}.}
        Set $k \gets \mathcal{O}(\log n \cdot n^{o(1)}), E_0 \gets E, F_0 \gets \emptyset, w_0 \gets w $. \\
        \For{$i = 1,\dots, k$}{ 
        $E_i \leftarrow \{e \in E_{i-1} : w_{i-1}(e)\notin \Z\} \setminus F_{i-1}$, $G_i=(V,E_i)$.\\
        \Repeat{there are no edges incident to a degree-1 node in $G_i$.}
        {
            \textit{Rake$(G_i)$:} Remove all edges incident to degree-1 nodes from $G_i$. \\
            \textit{Compress$(G_i)$:} Contract maximal paths of degree-2 nodes in $G_i$. \tcp*[r]{$\triangleright$ \Cref{lem:compress}}
        }
        Delete all isolated vertices and denote the pruned graph by $H_i$. \\
        Denote the set of edges removed by a \textit{Rake} operation by $R_i$. \\
        Compute a $(1/10,n^{o(1)})$-LDD $S_1, \dots, S_\ell$ on $H_i$. \tcp*[r]{$\triangleright$ \Cref{lm:LDD}}
        \For{all $j \in [\ell]$ in parallel}{
        Compute the set of bridges $B_{i,j}$ of $H_i[S_j]$. \tcp*[r]{ $\triangleright$ \Cref{lm:Bridges}}
        Compute a $(n^{o(1)},n^{o(1)})$-cycle cover $\mathcal{U}_{i,j}$ on $H_i[S_j]\setminus B_{i,j}$. \tcp*[r]{$\triangleright$ \Cref{lm:cycle_cover}}
        Compute a maximal set $\mathcal U_{i,j}'$ of non-overlapping cycles \tcp*[r]{$\triangleright$ \Cref{lem:MIS}}
        \For{all $C\in \mathcal U_{i,j}'$ in parallel}
        {
            $w_\mathrm{min}(C) \gets \min_{e \in C} \min \{w_{i-1}(e), 1 - w_{i-1}(e) \}$. \\
            Obtain $w_i$ by increasing or decreasing the edge weights on $C$ by $w_\mathrm{min}(C)$. 
        }
        }
        $F_{i} \gets F_{i-1} \cup R_i$.
    }
    \Return $w_k$
\end{algorithm}

\subparagraph*{Step 1: Compressing the Graph.}

As the first step in each iteration $i$, we prune the graph~$G_i$ by removing parts of it that cannot be covered by a cycle. To be precise, we simulate the following \textit{Rake} and \textit{Compress} procedures inspired by Miller and Reif \cite{MillerR89}:
\begin{itemize}
    \item \textit{Rake:} Remove all vertices of degree at most one and all incident edges.
    \item \textit{Compress:} Replace each maximal path $P = v_0,v_1,\dots,v_k$ such that $\mathrm{deg}(v_i) = 2$ for all $i = 1,\dots, k - 1$ with a single edge $(v_0,v_k)$.
\end{itemize}

Note that applying one iteration of \textit{Rake} and \textit{Compress} to an input graph $G$ always yields a \emph{topological minor} $H$ of $G$ as a result. 

\begin{definition}[Topological minor]
    We call a graph $H$ a \emph{topological minor} of $G$ if a subdivision of $H$ is isomorphic to a subgraph of $G$.
\end{definition}

We repeat the above operations to the current graph $G_i$ until we obtain a topological minor $H_i$ that does not contain any vertices of degree one or less.

\begin{restatable}{lemma}{RakeCompress}
    \label{lem:rake_and_compress}
    Let $G=(V,E)$ denote an $n$-node communication network of diameter $D$.
    There is a deterministic $\widetilde{\mathcal{O}}(\sqrt{ n } + D)$-round \CONGEST algorithm that iteratively simulates \textit{Rake} and \textit{Compress} on a subgraph $G' \subseteq G$ until the procedure terminates with a topological minor $H$ of $G$, where each connected component in $H$ is
    \begin{itemize}
        \item either a single cycle or 
        \item it has minimum degree at least three.
    \end{itemize}
    Further, we can simulate one round on $H$ in $\widetilde{\mathcal{O}}(\sqrt{ n } + D)$ rounds on $G$.
\end{restatable}

We defer the proof of the above lemma to Appendix~\ref{subsec:preliminaries}. 

\subparagraph*{Step 2: Computing a low diameter decomposition.} 

The goal of our next step is to round cycles in~$H_i$. Since our subroutines for making the graph bridgeless (see \Cref{lm:Bridges}) and for finding a cycle cover in a bridgeless graph (see \Cref{lm:cycle_cover}) take time proportional to the diameter, we start by computing a \emph{low diameter decomposition} (LDD) of~$H_i$. 

\begin{definition}
    \label{def:ldd}
    Given an undirected, unweighted graph $G = (V, E)$, a $(\beta, d)$-low diameter decomposition is a partition of $V$ into subsets $S_1, \dots, S_\ell$ such that
\begin{itemize}
    \item The diameter of $G[S_i]$ is at most $d$ for each $i\in [\ell]$;
    \item The number of edges with endpoints belonging to different pieces is at most $\beta \cdot |E|$.
\end{itemize}
\end{definition}

Low diameter decompositions are well studied in the \CONGEST model~\cite{MillerPX13,RozhonG20,ChangG21,ForsterGV21,GhaffariGR21,RozhonEGH22,GhaffariGHIR23,RozhonHG23}.
We use a deterministic version stated in \cref{lm:LDD}, which has slightly worse bounds than its randomized counterparts, but simplifies our analysis and does not incur a higher asymptotic runtime. 

\begin{lemma}[\cite{ChangS20}]\label{lm:LDD}
    There exists a deterministic algorithm that computes a $(1/10,n^{o(1)})$-low diameter decomposition in $n^{o(1)}$ rounds of \CONGEST.   
\end{lemma}

With~\cref{lm:LDD}, we partition the vertex set of $H_i$ into disjoint sets $S_1, \dots, S_\ell$ that induce clusters of small diameter with only few edges with endpoints in two different clusters.
For each cluster~$H_i[S_j]$, we carve out the set of bridges $B_j$ such that the remainder~$H_i[S_j] \setminus B_j$ is bridgeless. Using subroutines from existing work on edge-biconnectivity, we can compute the set of bridges efficiently enough for our use case.

\begin{lemma}[\cite{pritchard2005robust}]
    \label{lm:Bridges} Let $G$ be a graph of diameter $D$. There exists a deterministic algorithm that computes the set of bridges $B$ in $\Oo(D)$ rounds in \CONGEST.
\end{lemma}
\begin{proof}
    In Chapter 3 of~\cite{pritchard2005robust} a \CONGEST algorithm for edge-biconnectivity is presented. The objective of the edge-biconnectivity problem is to identify all bridges and biconnected components of a given graph. Thus, we can extract all bridges from this algorithm. Notably, the algorithm requires a spanning tree as input, but we can compute a BFS tree beforehand in $\Oo(D)$ rounds in \CONGEST.
\end{proof}

\subparagraph*{Step 3: Making each component bridgeless.} 

Using \cref{lm:Bridges}, we can compute the set of all bridges $B_{i,j}$ of $H_i[S_j]$, and the remainder~$H_i[S_j] \setminus B_{i,j}$ is a bridgeless graph, which we can cover by a union of (not necessarily disjoint) short cycles. Moreover, the diameter of (the connected components of) a graph does not increase after removing all its bridges.

\begin{lemma}
    \label{lm:bridgeremoval}
    Let $G=(V,E)$ be a connected graph of diameter $D$. Let $B$ be the set of bridges of $G$. Then the diameter of each connected component of $G'=(V,E\setminus B)$ is upper bounded by $D$.
\end{lemma}
\begin{proof}
    Let $v,w\in V$. We will show that $\mathrm{dist}_{G'}(v,w) = \mathrm{dist}_G(v,w)$ or $\mathrm{dist}_{G'}(v,w)= \infty$, i.e., vertices~$v,w$ end up in different connected components of $G'$. 
    Consider the shortest path between $v$ and $w$. Either it uses no bridge---which is unaffected by the removal of $B$, implying $\mathrm{dist}_{G'}(v,w) = \mathrm{dist}_G(v,w)$---or it uses a bridge $b\in B$ which is removed. Then there exists no other $v,w$ path by the definition of a bridge, which implies $\mathrm{dist}_{G'}(v,w)=\infty$.
\end{proof}

\subparagraph*{Step 4: Computing a large set of disjoint cycles.} 

Since each component is now bridgeless, all the remaining edges can be covered by cycles. We will round the edges on such a cycle. However, in the worst case, only one edge disappears in each iteration. Hence, we would like the cycles to be short. This is captured in the following \emph{cycle cover} definition.

\begin{definition}
    \label{def:cycle_cover}
    A \emph{$(d,c)$-cycle cover} is a collection of cycles $\cycles$ in a graph $G$ such that 
    \begin{enumerate}
        \item Each edge is contained in at least one cycle; 
        \item Each edge is contained in at most $c$ cycles; and 
        \item Each cycle is of length at most $d$. 
    \end{enumerate}
\end{definition}

Parter and Yogev~\cite{ParterY19} provide an algorithm that computes $(d,c)$-cycle covers in a bridgeless graph. Their result was later derandomized by Rozhoň, Grunau, Haeupler, Zuzic, and Li~\cite{RozhonGHZL22}, giving the following lemma. 

\begin{lemma}[\cite{RozhonGHZL22}]\label{lm:cycle_cover}
    Given a bridgeless $n$-vertex graph $G = (V, E)$ with diameter $D$, there exists a deterministic \CONGEST algorithm that
    computes a $(d, c)$-cycle cover $\cycles$ with $d = n^{o(1)} \cdot D$ and $c = n^{o(1)}$ in $n^{o(1)} \cdot D$ rounds.
\end{lemma}

As $H_i[S_j] \setminus B_{i,j}$ is bridgeless by construction, we can compute a cycle cover~$\mathcal U_{i,j}$ with \cref{lm:cycle_cover}.
To find a large set of \emph{disjoint} cycles that we can round in parallel, we compute a maximal independent set (MIS) on a virtual graph $\hat{G}$ of cycles, where two cycles are connected if they share an edge. 
The MIS algorithm from Censor-Hillel, Parter and Schwartzman~\cite{Censor-HillelPS20} combined with the network decomposition from Ghaffari, Grunau, and Rozhon~\cite{GhaffariGR21} gives the following result. 

\begin{lemma}[\cite{Censor-HillelPS20,GhaffariGR21}]\label{lem:MIS}
    Given an $n$-vertex graph $G = (V, E)$, there exists a deterministic \CONGEST algorithm that
    computes a maximal independent set in $\Oo(\log^5 n)$ rounds. 
\end{lemma}

As the last step of each iteration, we round all disjoint cycles in the MIS in parallel. 
Note that our rounding procedure guarantees that at least one edge in each cycle becomes integral.

\subsection{Full Algorithm and Proof}\label{subsubsec:proof_cycle_rounding}

Putting all the pieces together, we prove the correctness of our cycle rounding algorithm.

\begin{proof}[Proof of \cref{lem:cycle-elimination}]
Let $G_0$ denote the subgraph of the gift graph $G$ induced by all the edges~$e \in E$ with $w(e) \notin \mathbb{N}$.
Starting from $G_0$, we describe an iterative process that efficiently finds a large disjoint set of cycles in the current graph $G_i$ and rounds the weights of each cycle in parallel such that the weight of at least one edge on each cycle becomes integral.
Concurrently, we grow a set of edges $F_i$, which are not part of any cycle in the current residual graph $G_i$ and therefore can no longer be rounded to integer values by our procedure.
Initially, we set~$F_0$ to be empty.
We maintain the invariant that $F_i$ is a forest throughout our algorithm. Further, we show that after $k = \log m \cdot n^{o(1)}$ iterations, the graph $G_k$ is empty and the only edges with non-integral weights are contained in $F_k$, inducing a forest.
Each iteration can be performed in $\widehat{\mathcal{O}}(\sqrt{ n } + D)$ rounds of \CONGEST.
For a high-level overview of our algorithm, see \Cref{alg:cycle-rounding}.

\proofsubparagraph{Step 1: Compressing the graph.}
The goal of this step is to compress the graph $G_i$ to a denser version $H_i$ without long paths.
For this purpose, we iteratively apply the following two operations, which alternately remove all vertices of degree one and all vertices of degree two from $G_i$.
More precisely, we give an efficient implementation of the following two operations:
\begin{itemize}
    \item \textit{Rake:} Remove all vertices of degree at most one and all incident edges.
    \item \textit{Compress:} Replace each maximal path $P = v_0,v_1,\dots,v_\ell$ such that $\mathrm{deg}(v_i) = 2$ for all $i = 1,\dots, \ell - 1$ with a single edge $(v_0,v_\ell)$.
\end{itemize}
In addition, every time we compress a path $P$ into a single virtual edge $e_P$ we find $e_\mathrm{min} := \mathrm{argmin}_{e \in P} \min \{w(e), 1 - w(e) \}$ and set the weight of $e_P$ to be $w(e_P) := w(e_\mathrm{min})$.
Further, every edge $e \in P$ stores for both endpoints $v_0$ and $v_\ell$ of $P$ whether the number of original edges in $G$ between $e$ and $v_0$ and $v_\ell$ is even or odd.
After $\mathcal{O}(\log n)$ iterations of \textit{Rake} and \textit{Compress}, we arrive at a topological minor $H_i$ of $G_i$ such that each component in $H_i$ is either just a single cycle (of length two) or it has minimum degree at least three according to \Cref{lem:rake_and_compress}.
Each component that consists of just a single cycle immediately gets added to the set of disjoint cycles that we round at the end of the iteration.
Therefore, for the remainder of this proof we may assume that the minimum degree of $H_i$ is at least three.
The total round complexity across all iterations is $\widetilde{\mathcal{O}}(\sqrt{ n } + D)$.
Note that when we round a contracted edge, we can also forward the weight change to all edges on the path in $\Otilde(\sqrt n+D)$ rounds, for all contracted edges in parallel using \Cref{lem:rake_and_compress}.
Let $n_i$ and $m_i$ denote the number of vertices and edges in $H_i$, respectively.
Since $H_i$ has minimum degree at least three, we get that $m_i \geq 3/2 \cdot n_i$.
This step is crucial to ensure that the removal of up to $m_i/10$ edges with endpoints in two different clusters of the low diameter decomposition cannot remove all cycles from $H_i$ and we can still find enough cycles in each component of the LDD.

\proofsubparagraph{Step 2: Computing a low diameter decomposition.}
Next, we compute a low diameter decomposition of $H_i$ using \Cref{lm:LDD}, that decomposes $H_i$ into parts $S_1, \dots, S_\ell$ with $(n_i)^{o(1)}$ diameter and at most $m_i / 10$ edges between the clusters. 
Computing the low diameter decomposition takes $n^{o(1)}$ rounds on $H_i$ and according to \Cref{lem:rake_and_compress} each round on $H_i$ can be simulated in $\widetilde{\mathcal{O}}(\sqrt{ n } + D)$ rounds on $G$.
Thus we get a total runtime of $\widehat{\mathcal{O}}(\sqrt{ n } + D)$ for this step.
Algorithmically, we continue to work on each of the clusters $S_j$ in parallel---in the analysis, we consider all clusters simultaneously. 

\proofsubparagraph{Step 3: Making each component bridgeless.}
Using \Cref{lm:Bridges}, we compute the set of bridges $B_{i,1},\dots,B_{i,\ell}$ of $H_i[S_1],\dots,H_i[S_\ell]$ in parallel in $n_i^{o(1)}$ rounds.
From \Cref{lem:rake_and_compress} follows that this takes $n^{o(1)}\cdot \Otilde(\sqrt n+D)$ rounds on the communication network $G$. As a consequence, the graph $H_i[S_j] \setminus B_{i,j}$ is bridgeless.
Let $B_i := \bigcup_{j = 1}^{\ell} B_{i,j}$, then $B_i \subseteq H_i$ is a forest.

\proofsubparagraph{Step 4: Computing a large set of disjoint cycles.}
For each $H_i[S_j] \setminus B_{i,j}$, we compute a $(d,c)$-cycle cover $\cycles_{i,j}$ with $c=d=n^{o(1)}$ by using \Cref{lm:cycle_cover} on each connected component in $n^{o(1)}$ rounds, since the diameter of each connected component is bounded by $n^{o(1)}$ due to \Cref{lm:bridgeremoval}. Again by \Cref{lem:rake_and_compress} this takes $\Oo(n^{o(1)}\cdot(\sqrt n+D)) = \widehat{\mathcal{O}}(\sqrt n+D)$ rounds on the communication network $G$. 
To obtain a large set of disjoint cycles, we compute a maximal independent set on these cycles. 
For this purpose, we construct a virtual graph~$\widehat{G}$ of the cycles in $\cycles_{i,j}$, where two cycles are connected if they share an edge. 
Now we run a deterministic $O(\poly \log n)$-round \CONGEST algorithm that computes a maximal independent set, see \Cref{lem:MIS} on $\widehat{G}$.
Note that sending messages along the cycles incurs an additional multiplicative simulation overhead of $c\cdot d = n^{o(1)}$.
Combined with the $\widetilde{\mathcal{O}}(\sqrt{ n } + D)$ multiplicative overhead of simulating one round on $H_i$, this leads to an overall runtime of $\widehat{\mathcal{O}}(\sqrt{ n } + D)$ for this step.

Parter and Yogev~\cite{ParterY19} show that a maximal set of cycles covers an $\Omega(\tfrac{1}{cd})$-fraction of the edges. 
We repeat the argument here for completeness: 
Each cycle $C \in \cycles_{i,j}$ intersects at most~$|C| \cdot c$ cycles in $\cycles_{i,j}$. 
Thus, when we add the cycle $C$ into the collection $\mathcal D$ we cover $|C|$ many edges. 
When we remove the $|C| \cdot c$ cycles that intersect with $C$, we also remove the cover of at most $|C| \cdot c \cdot d$ many edges in $H_i$ that appear on these cycles. 
That is, at each iteration, the ratio between the number of edges we cover and the number of edges that the covering cycle removed is at least $1/dc$.

In order to round a single cycle $C$, we first traverse the cycle to identify the edge $e_\mathrm{min} \in C$ minimizing the minimum of $w_{i-1}(e)$ and $1 - w_{i-1}(e)$.
Let~$w_\mathrm{min}(C)$ denote this minimum value.
Next, starting from $e_\mathrm{min}$, we alternately increase and decrease the weights along the cycle by $w_\mathrm{min}(C)$ depending on whether the weight at $e_\mathrm{min}$ is closer to zero or one.
For each compressed edge $e_P$, every edge $e'$ on the corresponding path $P$ knows the parity of the number of edges between $e'$ and any endpoint of $P$ and can thus update its weight accordingly.
Whenever the weight of one edge in the path $P$ reaches either zero or one, we can safely remove the entire path $P$ as it can no longer be a part of any cycle in the remaining graph.
Note that even though the compressed graph $H_i$ might contain odd cycles, the corresponding cycle in $G$ obtained by unpacking all edges on compressed paths will always be even.
Further, all compressed edges can broadcast a message to all edges of their respective paths in parallel in $\widetilde{\mathcal{O}}(\sqrt{ n } + D)$ rounds on $G$.
Note that by definition, the updated weights $w_i$ remain in the interval $[0,1]$.
Whenever the weight of an edge reaches either zero or one, we remove it from~$H_i$.
Since the cycles are edge-disjoint, we can round all cycles in parallel. 
Furthermore, each cycle has a length of at most $n^{o(1)}$ in $H_i$ and thus this procedure takes $n^{o(1)}$ rounds on~$H_i$, which can be simulated in $\widehat{\mathcal{O}}(\sqrt{ n } + D)$ rounds on the communication network $G$, again using \Cref{lem:rake_and_compress}.

Further, we update $F_{i} := F_{i-1} \cup R_i$, where $R_i$ denotes the set of all edges in $G$ that got removed by any \textit{Rake} operation in this iteration.
Now we show that the set $F_{i}$ still induces a forest in $G$.
Note that for each edge $e \in R_i$, we know that $e$ does not lie on a cycle in $G_i$.
Now suppose there was an edge $e \in R_i$ that closed a cycle $C$ in $F_{i}$.
Clearly then this cycle must include at least one edge that is not included in $R_i$, since the set $R_i$ itself is acyclic.
From all those edges, we pick the edge $e'$ that was added to $F_j$ in the earliest iteration~$j < i$.
However, this implies that $C \subseteq G_{j-1}$.
Therefore, $e'$ cannot have been added to $F_j$ in iteration~$j$.
This completes the proof that $F_{i}$ remains indeed a forest.

Finally, we let $G_{i+1}$ be the graph induced by the set of edges $e$ such that $w_i(e) \notin \{0, 1\}$ without the edges in $F_i$.
Note that $G_{i+1}$ contains again the edges with endpoints in two different clusters in the low diameter decomposition, since their weights did not change in this iteration, as well as the set of bridges $B_i$.
This is because while each bridge $e \in B_{i,j}$ does not lie on a cycle in its cluster $H_i[S_j]$, it might still lie on a cycle in the graph $H_i$.

\proofsubparagraph{Quantifying the progress made in each iteration.}
We measure the progress that we make in each iteration by the number of edges $m_i$ remaining in the compressed residual graph~$H_i$.
Initially we have that $m_1 \leq m$.
The number of edges removed in each iteration is at least the number of disjoint cycles we construct in the previous step.
Observe that
\begin{equation*}
    \sum_{j = 1}^{\ell} \lvert E[ H_i[S_j] \setminus B_{i,j}] \rvert \geq \tfrac{9}{10} \cdot m_i - \lvert B_i \rvert.
\end{equation*} 
Since at least a $1/n^{o(1)}$-fraction of all edges in each cluster is covered by a cycle, we get that at least $(\tfrac{9}{10} \cdot m_i - |B_i|)/n_i^{o(1)}$ edges are covered by a cycle across all clusters $H_i[S_j] \setminus B_{i,j}$.
Since each cycle has a length of at most $n_i^{o(1)}$, this implies that at least $(\tfrac{9}{10} \cdot m_i - |B_i|)/n_i^{o(1)}$ edge weights are rounded to integral values in this iteration.
Since $B_i \subseteq H_i$ is a forest, we have that $\lvert B_i \rvert \leq n_i = \lvert V(H_i) \rvert $.
Our compression step using \textit{Rake} and \textit{Compress} guarantees that $m_i \ge \tfrac{3}{2}n_i$, or equivalently $n_i \le \tfrac{2}{3}m_i$ and hence 
\[
    \frac{9}{10}m_i - |B_i| \geq \frac{9}{10} m_i - n_i > \frac{1}{5} m_i.
\] 
As $G_{i+1}$ is a subgraph of $G_i$, it can be shown inductively that the compressed graph $H_{i+1}$ obtained by iteratively applying \textit{Rake} and \textit{Compress} is always a topological minor of $H_i$. Therefore, any edge in $H_{i+1}$ corresponds to at least one edge in $H_i$.
Furthermore, each edge from $H_i$ whose weight got rounded to zero or one in this iteration cannot appear in $H_{i+1}$ again.
Therefore, we have that 
\[
    m_{i+1} := |E(H_{i+1})| \leq m_i - \frac{1}{5 n_i^{o(1)}} \cdot m_i \leq m_i \cdot \left(1 - \frac{1}{n^{o(1)}}\right).
\]
Consequently, the number of remaining edges in $H_i$ always decreases by a factor of at least~$1/(1 - 1/n^{o(1)})$.
Thus, after $k = \log m \cdot n^{o(1)}$ iterations this procedure terminates with an empty graph $G_k$ and the only edges with non-integral weights remaining are the edges in the forest $F_k$.
\end{proof}

\newpage
\section{Distributed Santa Claus}
\label{sec:Santa_Claus}

Throughout this section, we show the following approximation for the Santa Claus problem.

\SantaApproximation*

We start by presenting the non-trivial LP formulation of the decision variant for the Santa Claus problem in \Cref{subsec:santa-lp-formulation}.
Then we compute an initial fractional assignment of gifts to children by applying our mixed packing and covering LP solver from \Cref{thm:simplified-mpc-solver}.
In \Cref{subsec:santa-integral-gift-assignment}, we show how to round this initial fractional solution to an integral solution using randomized selection and a procedure to round a forest efficiently---crucially, this requires two applications of the cycle rounding algorithm of \Cref{lem:cycle-elimination} from \Cref{subsec:cycle-elimination}.
Finally, we put all the pieces together in \Cref{subsec:santa-approx-factor} and show that the \CONGEST algorithm described throughout this section has \whp an approximation factor of $\alpha=\SantaApproxFactor$.

\subsection{Santa Claus LP}
\label{subsec:santa-lp-formulation}

Let~$\OPT$ denote the optimal value of the Santa Claus problem.
By a standard binary search framework, the following decision variant of the Santa Claus problem is equivalent to an $\alpha$-approximation algorithm for some $\alpha > 1$:
Given a threshold~$T \ge 0$, either compute a solution of value~$T/\alpha$ or determine that~$\OPT < T$.
Intuitively, the threshold~$T$ can be thought of as the optimal value that was found by binary search.
Throughout, we assume that~$T$ is known, and partition the set of gifts into \emph{big gifts}~$\cB := \{ g \in \cG : v_g \geq T/\alpha\}$ and \emph{small gifts}~$\cS := \{ g \in \cG : v_g < T/\alpha\}$. 
Originally introduced by \cite{LenstraST90}, this distinction has become the standard in the literature on Santa Claus.
Importantly, as we are showing an $\alpha$-approximation, any child that receives one big gift is satisfied and does \emph{not} require any more gifts.
Due to that insight, we implicitly assume that all big gifts have a value of~$T$ for the rest of the section.

In order to compute a fractional solution, we use the following non-trivial LP of Davies, Rothvoss and Zhang \cite{DaviesRZ20} (this expanded version was first used in~\cite{RohwedderRW25}) that we refer to as \LP throughout.
\begin{align}
    \sum_{s \in \cS: (c,s) \in E} v_{s}\cdot y_{cs} &\geq T \cdot \left(1 - \sum_{b \in \cB: (c,b) \in E} x_{cb}\right) &&\forall c \in \cC \label{eq:lp-objective-value} \\
    y_{cs} &\leq 1 - \sum_{b \in \cB: (c,b) \in E}^{} x_{cb} &&\forall c \in \cC, s \in \cS \label{eq:lp-gifts-sum-up-to-one} \\
    \sum_{c \in \cC}^{} x_{cb} &\leq 1 &&\forall b \in \cB \label{eq:lp-big-gift-assigned-to-one-child} \\
    \sum_{b \in \cB}^{} x_{cb} &\leq 1 &&\forall c \in \cC \label{eq:lp-child-gets-only-one-big-gift} \\
    \sum_{c \in \cC}^{} y_{cs} &\leq 1 &&\forall s \in \cS \label{eq:lp-small-gift-assigned-to-beta-many-children} \\
    x_{cb}, y_{cs}  &\geq 0 &&\forall c \in \cC, b \in \cB, s \in \cS  \label{eq:lp-nonnegativity}
\end{align}
During the following explanation of the constraints, it is helpful to think of an integral solution:
Constraints~(\ref{eq:lp-big-gift-assigned-to-one-child}) and~(\ref{eq:lp-child-gets-only-one-big-gift}) ensure that each big gift~$b$ is assigned to at most one child~$c$ and vice versa.
If child~$c$ is assigned a big gift~$b$, then $y_{cs} = 0$ and $c$ does not receive any small gifts as enforced by Constraint~(\ref{eq:lp-gifts-sum-up-to-one}).
Otherwise, we have~$y_{cs} \le 1$ and~$c$ is compensated in small gifts only.
The same argument applies to Constraint~(\ref{eq:lp-objective-value}), where the right-hand side is zero if~$c$ is assigned a big gift or~$T$ if it is not.
In the latter case, child~$c$ receives at least a value of~$T$ in small gifts.
Constraint~(\ref{eq:lp-small-gift-assigned-to-beta-many-children}) guarantees that each small gift is assigned at most once.

Additionally, Constraints~(\ref{eq:lp-objective-value}) and~(\ref{eq:lp-gifts-sum-up-to-one}) enforce the fundamental property that any child~$c \in \cC$ that receives small gifts with~$y_c \neq 0^{\abs{\cS}}$ is always assigned more than~$\alpha$ many gifts.
More formally, let $\cS_c := \{ s \in \cS: y_{cs} > 0 \}  \subseteq \cS$ denote the set of small gifts that~$c$ receives.
From the definition of the small gifts and Constraints~(\ref{eq:lp-gifts-sum-up-to-one}) and~(\ref{eq:lp-objective-value}) it follows that 
\begin{equation*}
    \sum_{s \in \cS_c} v_{s}\cdot y_{cs} 
    < \sum_{s \in \cS_c} T/\alpha\cdot \left(1 - \sum_{b \in \cB} x_{cb}\right) 
    = \frac{\abs{\mathcal{S}_c}}{\alpha} \cdot T \cdot \left(1 - \sum_{b \in \cB} x_{cb}\right)
    \leq \frac{\abs{\mathcal{S}_c}}{\alpha} \cdot \sum_{s \in \cS_c} v_{s}\cdot y_{cs}.
\end{equation*}
Consequently, any child~$c$ with~$y_c \neq 0^{\abs{\cS}}$ is assigned strictly more than~$\alpha$ many small gifts.
Thus, it is always possible to scale the small gift assignment of~$c$ to be of value~$T$, \ie we set
\begin{equation*}
    y'_c = \frac{1}{1 - \sum_{b \in \cB} x_{cb}} \cdot y_c \le 1^{\abs{\cS}}
\end{equation*}
such that~$\sum_{s \in \cS_c} v_{s}\cdot y'_{cs} \ge T$.
While~$y'$ is not necessarily a feasible solution to \LP anymore, we use this property to compute our approximation.

Given a fractional solution $(x,y) \in [0,1]^{\abs{\cC}\cdot\abs{\cG}}$ to \LP, the vectors $x \in [0,1]^{\abs{\cC}\cdot\abs{\cB}}$ and~$y \in [0,1]^{\abs{\cC}\cdot\abs{\cS}}$ denote the assignments for the big and small gifts, respectively.
In addition, let~$\Gx{x} = (\cC \cup \cB, E_x)$ denote the big gift graph between children~$c \in \cC$ and big gifts~$b \in \cB$ where an edge~$\set{c,b} \in E_x$ of weight~$x_{cb}$ is present if~$x_{cb} > 0$.
Similarly, we define the small gift graph~$\Gx{y} = (\cC \cup \cS, E_y)$ for~$c \in \cC$ and~$s \in \cS$ where each edge~$\set{c,s} \in E_y$ of weight~$y_{cs}$ is present if~$y_{cs} > 0$.

\subsection{Computing an Integral Gift Assignment}
\label{subsec:santa-integral-gift-assignment}

Starting from an (approximately) optimal fractional solution~$(x,y)$ to \LP, we show how to compute an integral gift assignment~$(X,Z)$ such that each child receives a value of~$T/\alpha$ throughout this section. As seen in \Cref{rem:finite_precision_congest}, we can assume the fractional input solution to be represented with precision $1/\poly n$. Throughout this section, rounding any intermediate value to precision $1/\poly n$ does not affect our $\Oo(\log n/ \log \log n)$ approximation guarantee. Therefore, we can represent any value with $\Oo(\log n)$ bits and can implement the algorithm in the \CONGEST model.

\subsubsection{Big Gift Clustering}
\label{subsubsec:santa-big-gift-clustering}

As a first step, we transform an (approximately) optimal fractional solution~$(x,y)$ to \LP into another feasible fractional solution~$(x',y)$ that corresponds to an acyclic subgraph~$\Gx{x'} \subseteq \Gx{x}$.
Similar to Bansal and Sviridenko~\cite{BansalS06}, our goal is to cluster children and big gifts in~$\Gx{x'}$ such that we can always find one child in the cluster to be satisfied in small gifts while all remaining children receive exactly one big gift.
Recall that any child is satisfied with receiving one big gift of value~$T$ and any non-zero small gift assignment can be scaled to a value of~$T$, as explained in \cref{subsec:santa-lp-formulation}.
We start by showing that removing all cycles from~$\Gx{x}$ still yields a feasible solution to \LP.

\begin{corollary}[Acyclic Big Gift Graph]
    \label{cor:santa-acyclic-big-gift-graph}
    Let $(x,y)$ be a feasible fractional solution to \LP.
    In {$\Ohat(\sqrt{ n } + D)$} rounds of \CONGEST, we can compute another feasible solution~$(x',y)$ such that the corresponding big gift graph $G_{x'} \subseteq G_x$ is a forest.
\end{corollary}

\begin{proof}
    Using \Cref{lem:cycle-elimination}, we obtain a sparse subgraph~$G_{x'} \subseteq G_x$ in $\Ohat(\sqrt{ n } + D)$ rounds of \CONGEST.
    From Property $(1)$ of \Cref{lem:cycle-elimination} follows that $\sum_{c \in \cC} x'_{cb} \leq 1$ and $\sum_{b \in \cB} x'_{cb} \leq 1$. 
    Thus, the Constraints~\eqref{eq:lp-big-gift-assigned-to-one-child} and~\eqref{eq:lp-child-gets-only-one-big-gift} are still satisfied. Since~$y$ does not change, the Constraints~\eqref{eq:lp-objective-value}, \eqref{eq:lp-small-gift-assigned-to-beta-many-children} and \eqref{eq:lp-gifts-sum-up-to-one} are also still satisfied.
    Hence, $(x',y)$ is a feasible solution to \LP.
\end{proof}

The following lemma uses similar arguments as Bansal and Sviridenko~\cite[Lemma 6]{BansalS06}.

\begin{lemma}[Big Gift Clusters]
    \label{lem:santa-big-gift-clusters}
    Let $(x',y)$ be the fractional assignment computed in \Cref{cor:santa-acyclic-big-gift-graph}.
    In $\Otilde(\sqrt{ n } + D)$ rounds of \CONGEST, we can compute a fractional assignment $(x^*,y)$ such that~$G_{x^\ast}$ is a forest where for any tree~$\mathcal{T}^\ast = (C^\ast \cup B^\ast, E^\ast)$ we have that
    \begin{enumerate}
        \item Each big gift $b \in B^\ast$ has degree at most two in $\mathcal{T}^\ast$;
        \label{prop:big_gift_degree}
        \item For any~$c \in C^*$, there always exists an integral assignment of big gifts $b \in B^\ast$ to the remaining $C^*\setminus\set{c}$ children; and
        \item If we cannot assign a big gift to every child $c \in C^\ast$, then the total value in small gift assignments in each~$\cT^*$ must be at least~$T/2$, \ie we have \label{prop:enough_small_gifts}
        \begin{equation}
        \label{eq:santa-big-gift-clusters-total-small-gift-value}
            \sum_{c \in C^\ast}\sum_{s \in \cS} v_s y_{cs} \ge T/2.
        \end{equation}
    \end{enumerate}
\end{lemma} 

\begin{proof}
    Let $\mathcal{T} = (C \cup B, E)$ be a tree in $G_{x'}$, then distinguish two cases in the pruning process.
    If $\mathcal{T}$ contains just one child $c \in C$, then Properties~1 and~2 are trivially satisfied.
    Since~$(x',y)$ is a feasible solution to \LP and~$c$ is isolated in $G_{x'}$, it receives a total value of~$T$ in small gifts alone and Property~3 is satisfied. 
    
    If $\mathcal{T}$ contains at least one big gift, then we prune~$\cT$ as follows:
    Using \Cref{lm:root_a_tree}, we root~$\mathcal{T}$ at an arbitrary child in $\widetilde{\mathcal{O}}(\sqrt{ n } + D)$ rounds and orient all edges in $\mathcal{T}$ towards the root.
    For each big gift $b \in B$ with $\mathrm{deg}(b) =: d > 2$, there is at most one child $c \in N(b)$ with~$x'_{cb} > 1 / 2$. We select $d - 2$ children~$c_1,\dots,c_{d-2}$ that are not the parent of $b$ in the tree such that~$x'_{c_i b} \leq 1 / 2$ for all $i = 1,\dots,d - 2$.
    For each selected child $c_i$, we remove the edge~$(c_i,b)$ from the graph which disconnects the components of its two endpoints.
    
    Let~$\mathcal{T}^\ast = (C^\ast \cup B^\ast) \subseteq \mathcal{T}$ be a tree after the pruning process, then there are two different cases.
    If~$\mathcal{T}^*$ contains a big gift~$b$ as a leaf, all children in~$C^*$ can receive a big gift. 
    Let~$c \in C^*$ be the only child neighboring~$b$ in~$\mathcal{T}^*$, then we can assign $b$ to $c$, root~$\mathcal{T}^*$ at $c$ and assign each big gift to its only \textit{child}.

    If all leaves are children, we show that there exists an integral assignment of the big gifts in $B^\ast$ to the children in $C^\ast \setminus \{ c \} $. 
    As each big gift has degree at most two, we can root $\mathcal{T}^\ast$ at $c$ and assign each big gift $b \in B^\ast$ to its only \textit{child}.
    Therefore, every child in $C^\ast \setminus \{ c \} $ can receive a big gift in this assignment and we have~$\lvert B^\ast \rvert = \lvert C^\ast \rvert - 1$.
    It remains to show Property~3.
    First, we show that at most one child in $C^\ast$ loses an edge during the pruning process. By construction, each child can only lose the edge to its parent gift. Let $c_1, c_2$ be two children that lost an edge during the pruning process, then the path between $c_1$ and $c_2$ traverses one of their parents. Thus, $c_1$ and $c_2$ cannot be in the same connected component after the pruning.
    Let $c^* \in C^*$ be the child that lost an edge to a big gift $b \in \mathcal{B} \setminus B^\ast$ if one such child exists.
    Since~$(x',y)$ is a feasible solution to \LP, it follows from Constraint~\eqref{eq:lp-objective-value} that
    \[
        \sum_{c \in C^\ast} \sum_{s \in \mathcal{S}} v_s \cdot y_{cs} \geq T \cdot \left( \lvert C^\ast \rvert - \sum_{c \in C^\ast} \sum_{b \in \mathcal{B}} x'_{cb} \right).    
    \]
    Using~$x'_{c^*b} \le 1/2$ and~$\sum_{c \in C^\ast} x'_{cb} \le 1$ for all $b \in B^\ast$, it holds that
    \[
        \sum_{c \in C^\ast} \sum_{b \in \mathcal{B}} x'_{cb} = x'_{c^*b} + \sum_{c \in C^\ast} \sum_{b \in B^\ast}  x'_{cb} \leq 1/2 + |B^\ast| = |C^\ast| - 1/2.
    \]
    Together, both equations yield \cref{eq:santa-big-gift-clusters-total-small-gift-value}.
\end{proof}

Importantly, even though $\mathcal{T}^\ast$ is a tree in the big gift graph $G_{x^*}$, any child $c \in C^*$ with a non-zero small gift assignment is still connected to more than~$\alpha$ many small gifts in~$\Gx{y}$.
Together with \cref{eq:santa-big-gift-clusters-total-small-gift-value}, this is crucial when randomly picking one child~$c \in C^\ast$ for each~$\mathcal{T}^\ast$ to be compensated in small gifts of value at least~$T/ \alpha$ in the next subsection.

\subsubsection{Small Gift Assignment}
\label{subsubsec:santa-small-gift-assignment}

In the following, we compute an integral small gift assignment~$Z$ of value $T/\alpha$ for exactly one child~$c$ in each tree~$\cT^*$ computed in \cref{lem:santa-big-gift-clusters}.
This automatically implies an integral big gift assignment~$X$ by rooting all trees~$\cT^*$ in their respective roots~$c$ and assigning all big gifts to their child nodes (see \cref{lem:santa-actual-big-gift-assignment}). 

First, we show that we can randomly pick exactly one child~$c$ in each tree~$\cT^*$ to be satisfied in a fractional small gift assignment.
We use the following Chernoff bound from~\cite[Corollary 13]{Kuszmaul25} to bound the number of times that one small gift is fractionally assigned to a child.

\begin{lemma}[Chernoff Bound]
    \label{lem:chernoff}
    Let $X_1,\dots,X_n \in [0,1]$ be independent random variables with means $p_1,\dots,p_n$. Let $\mu = \sum_{i = 1}^{n} p_i$ and $X = \sum_{i = 1}^{n} X_i$, then for any $r \geq 2$ with $1 \leq r \mu$
    \[
        \mathbb{P}[X \geq r\mu] \leq \mathcal{O}(1 / r)^{\Omega(r \mu)}.
    \]
\end{lemma}

The most important insight to the following lemma is that the random selection happens independently and in parallel for all~$\cT^*$ such that all respective decisions whether a small gift is assigned to a tree are in fact \emph{independent}.
More formally, we show the following lemma.

\begin{lemma}[Randomized Selection]
    \label{lem:santa-randomized-selection-of-small-gifts}
    Let~$(x^*,y)$ be the fractional assignment computed in \cref{lem:santa-big-gift-clusters}.
    In $\widetilde{\mathcal{O}}(\sqrt{ n } + D)$ rounds of \CONGEST, we can transform~$y$ into a fractional small gift assignment~$z$ such that for some~$\beta =$ \SantaApproxFactor\, it holds \whp that
    \begin{enumerate}
        \item For all $\mathcal{T}^\ast$, either all children in $\mathcal{T}^\ast$ each receive one big gift or there is one child $c \in \mathcal{T}^\ast$ such that $\sum_{s \in \cS}v_{s}\cdot z_{cs} \ge T/\beta$,
        \item $\sum_{c \in \cC}^{} z_{cs} \leq 1$ for all~$s \in \cS$.
    \end{enumerate}
\end{lemma}

\begin{proof}
    We write $(C_1, B_1), (C_2, B_2), \dotsc, (C_p, B_p)$ to denote the collection of disjoint big gift clusters that~$x^*$ induces.
    If~$\lvert B_i \rvert \geq \lvert C_i \rvert $, then all children in $C_i$ receive exactly one big gift and no small gifts.
    Otherwise, if $\lvert B_i \rvert  = \lvert C_i \rvert - 1$, then we pick one child to be satisfied with small gifts.
    For all~$c \in C_i$, let~$V_\mathcal{S}(c) := \sum_{s \in \mathcal{S}} v_s y_{cs}$ denote the total value that~$c$ receives in small gifts.
    By Property~(\ref{prop:enough_small_gifts}) of \cref{lem:santa-big-gift-clusters}, we have that~$V_\mathcal{S}(C_i) := \sum_{c \in C_i} V_\mathcal{S}(c) \geq T / 2$.

    For each~$(C_i, B_i)$, we first define a probability distribution and then randomly pick exactly one child in $C_i$.
    The weight of each child corresponds to its relative value of small gifts that it contributes to the cluster.
    More precisely, for each $c \in C_i$, we set $p_c := V_\mathcal{S}(c) / V_\mathcal{S}(C_i)$.
    Clearly, this definition satisfies $\sum_{c \in C_i} p_c = 1 $ for all $i \in [p]$.

    \proofsubparagraph{\boldmath\CONGEST implementation.}

    Next, we show how we can implement this random process for each cluster in parallel.
    We implement a \textit{Rake} and \textit{Compress} style procedure that decomposes each cluster $(C_i,B_i)$ in parallel.
    Since $\lvert B_i \rvert = \lvert C_i \rvert - 1$, every leaf in the cluster has to be a kid. Thus, due to property \ref{prop:big_gift_degree} in \Cref{lem:santa-big-gift-clusters} we get that every gift $b \in B_i$ has degree two.
    Hence, we may contract all paths $P = (c, b, c')$, where $b \in B_i$ and $c, c' \in C_i$ to the single edge $\{ c,c' \}$. Thus, we create a virtual tree $\cT$ that contains only the children in $C_i$.
    Now we iteratively apply the following two operations to $\cT$.
    Define $\cT_0 := \cT$ and $w_0: V(\cT) \to [0,1]$ by $w_0(c) = p_c$ and obtain $(\cT_{j+1},w_{j+1})$ from $(\cT_j,w_j)$ as follows:
    \begin{itemize}
        \item \textit{Rake:} Remove each child $c \in C_i$ of degree one from $\cT_j$ and send $w_j(c)$ to its only neighbor in $\cT_j$.
        \item \textit{Compress:} Contract each path $P$ of at least two degree-2 vertices into a single vertex $c$ and gather the weights of all vertices in $P$ at $c$.
        \item For each surviving vertex $c$, set $w_{j+1}(c)$ as $w_j(c)$ plus the sum of all weights sent to $c$.
    \end{itemize}
    First, we observe that after $\mathcal{O}(\log n)$ iterations of \textit{Rake} and \textit{Compress} we have decomposed the entire tree \cite{CP19}. Let $V_{j,1} := \{ v \in V(\cT_j): \mathrm{deg}(v) = 1 \}, V_{j,2} := \{ v \in V(\cT_j): \mathrm{deg}(v) = 2 \}  $ and $V_{j,3} := \{ v \in V(\cT_j): \mathrm{deg}(v) \geq 3 \} $.
    If $\lvert V_{j,2} \rvert < 3 / 5 \cdot \lvert V(\cT_j) \rvert $, then due to $\lvert V_{j,1} \rvert > \lvert V_{j,3} \rvert$ we have that $\lvert V_{j,1} \rvert > 1 /5 \cdot \lvert V(\cT_j) \rvert $ and thus a constant fraction of vertices are removed in a single iteration.
    Otherwise, if $\lvert V_{j,2} \rvert  \geq 3 / 5 \cdot \lvert V(\cT_j) \rvert $, then there are at most $\lvert V_{j,1} \rvert + \lvert V_{j,3} \rvert - 1$ connected components in the subgraph of $\cT$ induced by $V_{j,2}$.
    Since at most one vertex survives in every connected component, this implies that at least $1 / 5 \cdot \lvert V(\cT_j) \rvert $ vertices are removed from $\cT_j$.
    Further, we observe that due to \Cref{lem:compress} each iteration can be implemented in $\widetilde{\mathcal{O}}(\sqrt{ n } + D)$ rounds of \CONGEST. 
    
    Now we show how to use this decomposition to efficiently sample a vertex from every cluster in parallel.
    Let $\cT_k$ denote the last non-empty tree computed by the above procedure. Note that every vertex of $\cT_k$ must have been removed by the \textit{Rake} operation and thus~$\cT_k$ contains at most two vertices.
    We first sample a vertex from $\cT_k$ according to the probability distribution~$w_k$.
    Now we show how to use repeated subsampling to sample a vertex from $\cT$ according to the probabilities $p_c$ for $c \in C_i$.
    Suppose that we have sampled a vertex $v_{j+1}$ from~$\cT_{j+1}$.
    Let $\text{Compress}_j$ be the (possibly empty) set of vertices that were contracted into vertex~$v_{j+1}$ by the \textit{Compress} operation on $\cT_j$.
    If $\text{Compress}_j \neq \emptyset$, then $v_{j+1}$ picks exactly one vertex from $\text{Compress}_j \cup \{ v_{j+1} \} $ according to the probability distribution $p_j(v) := w_j(v) / w_{j+1} (v_{j+1})$. Then it broadcasts the random decision to all vertices in $\text{Compress}_j$ within $\widetilde{\mathcal{O}}(\sqrt{ n } + D)$ rounds using \Cref{lem:compress}.
    Denote the picked vertex as $v_j^c$.
    Next, let $\text{Rake}_j$ be the (possibly empty) set of vertices neighboring $v_j^c$ in $\cT_j$ that were removed from $\cT_j$ due to the \textit{Rake} operation. Again, $v_j^c$ randomly picks a vertex from $\text{Rake}_j \cup \{ v_j^c \}$ with probability proportional to their weights in $w_j$ and broadcasts this decision to all vertices in $\text{Rake}_j$.
    This takes one round on $\cT_j$ and according to \Cref{lem:compress} one round on $\cT_j$ can be simulated in $\widetilde{\mathcal{O}}(\sqrt{ n } + D)$ rounds on $G$.
    After $k$ iterations, this yields a randomly sampled vertex from $\cT$ according to probability distribution $w_0(c) = p_c$.
    Since $k = \mathcal{O}(\log n)$ the repeated subsampling procedure takes $\widetilde{\mathcal{O}}(\sqrt{ n } + D)$ rounds in total.

    \proofsubparagraph{Scaling Small Gifts.}

    Let~$\cC_s = \{ c_i: i = 1,\dots,p\} \subseteq \cC$ denote the set of children chosen to receive small gifts where 
    each~$c_i$ uniquely identifies cluster pair~$i \in [p]$.
    We define a new fractional small gift assignment $z$ as follows:
    \[
        z_{cs} := 
            \begin{cases}
                y_{cs} \cdot T / V_\mathcal{S}(c), & \text{if } c \in \mathcal{C}_s, \\
                0, & \text{otherwise}.
            \end{cases}
    \]
    Using Constraints~$(\ref{eq:lp-objective-value})$ and~$(\ref{eq:lp-gifts-sum-up-to-one})$ of \LP, we get for all $s \in \mathcal{S}$ and $c \in C_s$ that
    \begin{equation}
        \label{eq:feasibility_of_scaled_values}
        z_{cs} \leq y_{cs} \cdot T / V_\mathcal{S}(c) \leq T / V_\mathcal{S}(c) \cdot \left( 1 - \sum_{b \in \mathcal{B}} x'_{cb} \right) \leq \frac{\sum_{s \in \mathcal{S}} v_s \cdot y_{cs}}{V_\mathcal{S}(c)} = 1. 
    \end{equation}
    Further, it holds that~$\sum_{s \in \cS} v_s z_{cs} = T$ for all children~$c \in \cC_s$. 

    As a last step, we show that the random choice of~$z$ (fractionally) assigns each small gift to at most~$\beta$ many children with high probability.
    For each small gift~$s \in \cS$ and cluster pair~$i$, define the random variable~$Z_{is} \in [0,1]$ by $Z_{is} := z_{c_is}$.
    For a fixed gift~$s \in S$, the~$p$ many random variables~$\set{ Z_{is} : i \in [p]}$ are independent, because the random selection of children~$c_i$ happens independently in all~$p$ clusters.
    For each small gift~$s \in \cS$, we have that
    \begin{equation*}
        \E \left[  \sum_{i \in [p]} Z_{is} \right] 
        = \sum_{i \in [p]} \E \left[  Z_{is} \right] 
        = \sum_{i \in [p]} \sum_{c \in C_i} p_{c} \cdot z_{cs}.
    \end{equation*}
    Recall that $p_c = V_\mathcal{S}(c) / V_\mathcal{S}(C_i)$. Then we get that
    \begin{equation}
        \label{eq:santa-expected-value-of-rvs}
        \E \left[  \sum_{i \in [p]} Z_{is} \right] 
        = \sum_{i \in [p]} \sum_{c \in C_i} p_{c} \cdot z_{cs}
        = \sum_{i \in [p]} T \cdot \frac{\sum_{c \in C_i} y_{cs}
        }{V_\mathcal{S}(C_i)}
        \le 2 \sum_{c \in \cC} y_{cs}
        \le 2,
    \end{equation}
    where the last inequality follows from Constraint \eqref{eq:lp-small-gift-assigned-to-beta-many-children} of \LP.
    For a sufficiently large constant $C$, let~$\beta = C \cdot \log n / \log \log n$. Then using \Cref{lem:chernoff} it holds that
    \begin{align*}
        \label{eq:santa-chernoff-bound-on-number-children}
        \P \left[  \sum_{i \in [p]} Z_{is} > \beta \right] 
        &\leq \P \left[  \sum_{i \in [p]} Z_{is} > \beta / 2 \cdot \mathbb{E}\left[\sum_{i \in [p]} Z_{is} \right] \right]
        \leq \mathcal{O}(2 / \beta)^{\Omega(\beta)} \\
        &= \mathcal{O}( \log \log n / \log n)^{\Omega(\log n / \log \log n)}
        = 2^{-\Omega(\log n)}.
    \end{align*}
    Further, we have that
    \begin{equation*}
        \sum_{i \in [p]} Z_{is} = \sum_{i \in [p]} z_{c_i s} = \sum_{c \in C} z_{cs}.
    \end{equation*}
    Using a union bound over all small gifts, we see that each small gift is \whp fractionally assigned at most~$\beta$ times.
    As a last step, each child~$c \in \cC_s$ locally scales down their small gift assignment by~$\beta$ such that~$z_c \in [0,1]^\cS$.
    Consequently, each child receives \whp a total value of~$T/ \beta$ and each small gift is \whp fractionally assigned at most once.
\end{proof}

Intuitively, the parameter~$\beta$ measures the distance to the optimal solution where each small gift is assigned at most once. Therefore, it is closely related to our approximation factor by~$\alpha = 4 \beta$.
As the last step, we round~$z$ to an integral assignment~$Z$ of small gifts, such that no gift is assigned to multiple children. 
The high-level idea of the rounding can be described as follows:
For a fixed child~$c \in \cC_s$, let $k_c = \sum_{s \in \cS} z_{cs}$ denote the (fractional) number of small gifts~$c$ receives in solution~$z$.
Next we show how to round this fractional solution such that each child receives~$\floor{k_c}$ small gifts.
Hence, each child loses the value of at most one small gift.
In the worst case, this is the most valuable gift of value $\max_{s \in \cS} v_s$ and each child receives a value of at least~$T / \beta - \max_{s \in \cS} v_s$.
More formally, we show the following lemma.

\begin{lemma}[Integral Assignment]
\label{lem:santa-rounding-of-small-gifts}
    Let~$z$ be the fractional small gift assignment in \cref{lem:santa-randomized-selection-of-small-gifts}.
    In $\Ohat(\sqrt n+D)$ rounds, we can compute an integral small gift assignment~$Z$ such that \whp
    \begin{equation}
        \label{eq:santa-integral-small-gift-assignment-value}
        \forall c \in \cC_s: \sum_{s \in \cS} v_{s}\cdot Z_{cs} \geq \frac{T}{\beta} - \max_{s \in S} v_s.
    \end{equation}
\end{lemma}
\begin{proof}
    We prove the lemma in two steps: First we remove all cycles in~$G_z$ to obtain the forest~$G_{z'} \subseteq G_z$, then we integrally assign the gifts for each tree of~$G_{z'}$ such that each child loses at most one gift.
    
    \proofsubparagraph{Algorithm Part 1.} 
    \Cref{lem:cycle-elimination} shows how to remove cycles in a weighted graph. 
    We cannot apply this lemma directly to $G_z$, since we have two types of weights: we have the fractional weights on the edges $z_{cs}$, but we also have the weights of the small gifts $v_s$. We apply \Cref{lem:cycle-elimination} to $G_z$ with an auxiliary weight function $w$ and two small adjustments. First, we describe this weight function. 
    In order to ensure that the cycle removal maintains the total value of~$T/\beta$ for each child, we set the edge weights to~$w_{cs}:=z_{cs} \cdot v_s$ for all~$c \in \cC_s$ and~$s \in \cS$. 
    These weights are no longer in $[0,1]$ as in \Cref{lem:cycle-elimination}. We give the two adjustments to the algorithm. First, to round a cycle $C$, we compute~$w := \min_{c \in C}\min\set{w_{cs},v_s-w_{cs}}$ and alternately change the edge weights to $w_{cs}' = w_{cs}\pm w$. At the same time we update $z'_{cs} = w'_{cs} / v_s$ for all $c \in \mathcal{C}_s$ and $s \in \mathcal{S}$.
    By the choice of $w$, at least one edge in each cycle becomes integral with~$z_{cs}' \in \set{0,1}$ while we have~$z_{cs}'\in [0,1]$ for all other edges.    
    This leads us to the second adjustment: we remove an edge $\{ c, s \} $ from $G_z$ if $w_{cs}\in \{0,v_s\}$, which corresponds to~$z_{cs}\in \{0,1\}$ being integral. Note that removing an edge~$z_{cs} = 1$ means that the small gift~$s$ is permanently assigned to child~$c$ and $z_{cs} = 0$ means the child forgoes the gift. 
    
    These two adjustments do not affect correctness or the running time in \Cref{lem:cycle-elimination}, hence we eliminate cycles from~$G_z$ in $\Ohat(\sqrt n +D)$ rounds. Finally, we obtain the edge weights in~$G_{z'}$ by scaling $z_{cs}' = w'_{cs}/v_s$ for all~$c \in \cC_s$ and~$s \in \cS$. We are left with the acyclic graph~$G_{z'}$.

    \proofsubparagraph{Algorithm Part 2.}
    Next, we show how to round the remaining forest $G_{z'}$. 
    For each tree~$\cT$ in~$G_{z'}$, we permanently assign small gifts~$s \in \cS$ to children~$c \in \cC_s$ starting from the leaves.
    We root $\mathcal{T}$ at an arbitrary child $c \in \mathcal{T}$. Then, each gift is assigned to its parent in the tree. Since each child has at most one parent in $\mathcal{T}$, it loses at most~$\max_{s \in S} v_s$ in value and thus \cref{eq:santa-integral-small-gift-assignment-value} holds.

    \proofsubparagraph{Round Complexity.} 
    A naive distributed implementation of this algorithm takes diameter($\mathcal{T}$) rounds, which can be as large as $n-1$. Therefore, we show how to implement this algorithm in~$\Ohat(\sqrt n+D)$ rounds. 

    The main observation here is that on a path, all nodes are determined by what happens on the endpoints. 
    This can be seen as follows: consider a path that ends in a degree 1 vertex. 
    If this is a gift, then all odd edges on the path are rounded up to 1 and all even edges are deleted. 
    If the last vertex is a child, then all odd edges are deleted and all even edges are rounded up to 1. 
    Hence, each edge only needs to know 1) is the endpoint of the path a gift or a child and 2) are there an odd or even number of edges till this endpoint. 
    Hence, all paths can be contracted and unrolled once the endpoints are known.
    We perform this with a version of \textit{Rake} and \textit{Compress}; see \Cref{lm:root_a_tree}: 
    \begin{itemize}
        \item \emph{Rake:} Decide on all edges incident to a vertex of degree one. 
        \item \emph{Compress:} Contract all paths into a single edge, storing on the edge whether the path contains an odd or even number of edges. 
    \end{itemize}

    By the same proof as \Cref{lm:root_a_tree}, the graph is empty after $\Oo(\log n)$ iterations and each iteration can be performed in $\Ohat(\sqrt n+D)$ rounds. 
\end{proof}

We conclude with the following lemma that shows how to efficiently compute an integral big gift assignment~$X$.

\begin{lemma}
    \label{lem:santa-actual-big-gift-assignment}
    Let~$x^*$ be the fractional big gift assignment in \cref{lem:santa-big-gift-clusters} and~$Z$ the integral small gift assignment in \cref{lem:santa-rounding-of-small-gifts}.
    In $\widetilde{\mathcal{O}}(\sqrt{ n } + D)$ rounds of \CONGEST, we can compute a big gift assignment~$X$ for all~$c \in \cC$ with~$Z_c = 0^{\abs{\cS}}$.
\end{lemma}
\begin{proof}
    In~$\Otilde (\sqrt n+D)$ rounds, we root all trees~$\cT^*$ where there is one child that cannot be satisfied with a big gift in parallel at their respective children~$c^*$ with~$Z_{c^*} \neq 0^{\abs{\cS}}$ using \Cref{lm:root_a_tree}.
    By \Cref{lem:santa-big-gift-clusters}, each big gift in~$\cT^*$ has degree at most~$2$.
    Thus, we integrally assign each child its parent big gift.
    Hence, we computed a big gift assignment~$X$ in $\Otilde (\sqrt n+D)$ rounds. 
\end{proof}

\subsection{Proof of the Approximation Factor}
\label{subsec:santa-approx-factor}

We conclude the previous subsections with the proof of \cref{thm:santa-approximation}.

\begin{proof}[Proof of \cref{thm:santa-approximation}]
    Without loss of generality, we assume that the values of gifts $v_g$ are integral. This can be seen as follows: round down the real values $v_g$ to the precision~$1/f$ for some $f= \Oo(\poly n)$ in order to become fractional. Then multiply all values $v_g$ by the same factor $f$ such that all values become integral. At the end, we obtain the solution by dividing by the factor $f$ again.

    Let $\beta =$ \SantaApproxFactor{} be as in \cref{lem:santa-randomized-selection-of-small-gifts} and~$\alpha = 4\beta$.
    Throughout this proof, we explicitly state which~$\alpha$ is used in the LP formulation of Santa Claus by writing~$\LPext{\alpha}{T}$.
    We compute an $\alpha$-approximation in two steps:
    First, we determine a value~$T \ge \OPT$ via global binary search and repeatedly approximately solving $\LPext{\alpha}{T}$ for different values of~$T$. 
    Recall that~$\OPT$ denotes the optimal (integral) value of the Santa Claus problem.
    Second, we carefully round an approximate solution~$(x,y)$ into an integral solution~$(X,Z)$ that is an \SantaApproxFactor{}-approximation.
    
    \proofsubparagraph{Find Optimal Value and Fractional Solution.} Let~$\eps = 1/2$ and~$V = \sum_{g \in \mathcal{G}} v_g$ denote the total sum of gift values.
    In~$\Oo(D)$ rounds, we elect a leader that coordinates the binary search.
    The goal of our binary search is to find an integer value~$T$ such that~$\LPext{\alpha/2}{T/2}$ is feasible and~$T \ge \OPT$.
    Let~$a,d \in \set{0, \dotsc, V}$ with~$a \le d$, then we maintain the following invariant for our search interval~$[a,d]$: $\LPext{\alpha/2}{a}$ is feasible and~$d \ge \OPT$.
    This is trivially satisfied in the beginning, as $\LPext{\alpha/2}{0}$ is always feasible and~$d = V \ge \OPT$.
    
    In each of the following steps, the leader communicates the next value~$T = \floor{(a + d)/2}$ to all other nodes in $\Oo(D)$ rounds.
    Based on~$T$ and $\alpha$, all gift nodes can locally distinguish whether they belong to the set of big gifts~$\cB := \{ g \in \cG : v_g \geq T/\alpha\}$ or the set of small gifts~$\cS := \{ g \in \cG : v_g < T/\alpha\}$. 
    Observe that the gift size distinction is the same for both $\LPext{\alpha}{T}$ and $\LPext{\alpha/2}{T/2}$, thus they only differ in Constraint~$(\ref{eq:lp-objective-value})$.
    Next, we try to compute a feasible fractional solution~$(x,y)$ to~$\LPext{\alpha/2}{T/2}$ using~\cref{thm:solving_mixed_packing_covering}.
    Depending on the outcome, the leader adjusts the search interval to~$[T, d]$ if $\LPext{\alpha/2}{T/2}$ is feasible and~$[a, T]$ otherwise.
    Once $d - a = 1$, we run one last test with value $T = d$, which finally determines the true value of~$T$ and allows us to terminate the search.
    Note that by \cref{thm:solving_mixed_packing_covering}, our solution~$(x,y)$ only undershoots the covering Constraint~$(\ref{eq:lp-objective-value})$ such that for all~$c \in \cC$, it holds that
    \begin{equation*}
        \sum_{s \in \cS} v_{s}\cdot y_{cs} \geq \frac{T}{2} \cdot \left(1 - \sum_{b \in \cB} x_{cb}\right).
    \end{equation*}
    Since binary search takes at most~$\Oo(\log V)=\Oo(\log n)$ steps\footnote{Here we use that the original values $v_g\in \R$ are bound by $\poly n$. Hence after rounding and scaling they are still bound by $\poly n$, and also $V=\sum_{g\in \cG}v_g=O(\poly n)$. Hence $\log V=O(\log n)$.} to find the integer $T$, where each step takes $\widetilde\Oo(D + D \cdot \log^3 n)$ rounds due to \Cref{thm:solving_mixed_packing_covering}, we can find~$T \ge \OPT$ and a feasible fractional solution~$(x,y)$ to~$\LPext{\alpha/2}{T/2}$ in~$\widetilde\Oo(D \cdot \log^4 n )$ rounds of \CONGEST.

    \proofsubparagraph{Round to Integral Solution.} 
    Based on~$(x,y)$, we compute a series of fractional assignments, where all transformations take~{$\Ohat(\sqrt{ n } + D)$} rounds.
    With \cref{cor:santa-acyclic-big-gift-graph}, we obtain a fractional solution~$(x',y)$ such that~$\Gx{x'}$ is acyclic.
    Using \cref{lem:santa-big-gift-clusters}, we compute a fractional assignment~$(x^*,y)$ where~$G_{x^*}$ is a forest of highly structured trees $\cT^*$.
    Based on~$y$, we use \cref{lem:santa-randomized-selection-of-small-gifts} to compute a fractional small gift assignment~$z$ such that all children in~$\cC_s$ receive \whp a total value of~$T/(2\beta)$ in small gifts.
    As a last step, with \cref{lem:santa-rounding-of-small-gifts}, we obtain an integral small gift assignment~$Z$  where each child~$c \in \cC_s$ receives \whp at least a value of
    \begin{equation*}
        \frac{T}{2\beta} - \max_{s \in \cS} v_s \ge \frac{2T}{\alpha} - \frac{T}{\alpha} = \frac{T}{\alpha},
    \end{equation*}
    where the inequality follows from the definition of small gifts~$\cS$ and~$\alpha$.
    By \cref{lem:santa-actual-big-gift-assignment}, we compute an integral big gift assignment~$X$ where each child receives exactly one big gift of value at least~$T/ \alpha$.
    Hence, we computed \whp an \SantaApproxFactor{}-approximation. Since all lemmas used to transform the fractional assignments into an integral assignment take~$\Ohat (\sqrt{n} + D)$ rounds, we obtain our approximation in~$\Ohat (\sqrt{n} + D)$ rounds of \CONGEST.
\end{proof}

\newpage
\section{Lower Bounds}
\label{sec:lower-bounds}

Throughout this section, we prove several hardness results for the Santa Claus problem.  
In \cref{subsec:lower-bounds-congest}, we show a lower bound of $\widetilde{\Omega}(\sqrt{n}+D)$ rounds to decide whether a solution of non-zero value exists in the \CONGEST model. 
This implies that computing any (multiplicative) approximation to the integral Santa Claus problem requires the same lower bound.
In \cref{subsec:lower-bounds-local}, we use a simpler version of this construction to show a diameter lower bound for the same problem in the \LOCAL model.
Finally, we show that exactly solving the fractional version of the Santa Claus problem requires at least diameter rounds in any distributed setting in \cref{subsec:lower-bounds-augementing-paths}.  

\subsection{\boldmath\CONGEST Model}
\label{subsec:lower-bounds-congest}

We start out by showing our following main hardness result.

\ThmSantaClausLowerBound*

Motivated by the construction of Peleg and Rubinovich \cite{doi:10.1137/S0097539700369740}, we construct a family of graphs $SC_n$ with $\Oo(n)$ nodes and $\Oo(\log n)$ diameter.
Then we show that solving the Santa Claus problem on $SC_n$ is as hard as solving the set-disjointness problem on sets of size $\Oo(\sqrt{n})$ in the two-party model. 

\paragraph{Construction.}
Intuitively, we construct a graph that consists of $\sqrt{n}$ disjoint paths of length~$\Oo(\sqrt{n})$, where each path is its own Santa Claus instance of value $0$ or $1$ depending on the input for Alice and Bob. To reduce the diameter of the graph to $\Oo(\log n)$, we add a binary tree of height $\Oo(\log n)$ on top of the paths, allowing for shortcuts between the endpoints of the graph. Then we show that the small bandwidth of the binary tree implies that we still require $\widetilde{\Omega}(\sqrt{n}+D)$ rounds to solve the Santa Claus problem on all paths.

More formally, let $n\in \mathbb{N}$ such that $\sqrt{n} \in \mathbb{N}$. Let $SC_n$ be a graph that consists of $\sqrt{n}$ disjoint paths~$P_i$, each of which consists of~$2\sqrt{n}-1$ vertices, where each path alternates between gifts of value 1 and children. All gifts in $SC_n$ are connected by a full binary tree $T$ with $\sqrt{n} - 1$ many leaves $\ell_1,\dots,\ell_{\sqrt{n} - 1}$ that are children. Each leaf $\ell_j$ connects to the $j$-th gift $g_{j}^{i}$ of every path $P_i$. To ensure a Santa Claus instance, the levels of $T$ are alternating between children and gifts, e.g., the parents of the leaves are gifts. In particular, the binary tree cannot interfere with the solution of the paths or change the overall optimal solution size. Therefore, we add a private gift $p_c$ to each child vertex~$c^T\in V(T)$. 
Let $c_j^i$, and $g_j^i$ denote children and gifts of path $P_i$, respectively.
Formally, we have $\sqrt{n}$ copies $P_1, \dots, P_{\sqrt{n}}$ of the path $\mathcal{P}$:
\begin{align*}
    V(\mathcal{P}) &=\{c_{j} \mid j=1,\dots,\sqrt{n}\} \cup \{g_{j} \mid j=1, \dots, \sqrt{n}-1\}, \\
    E(\mathcal{P}) &= \{\{c_{j},g_{j}\} \mid j=1,\dots,\sqrt{n} - 1\} \cup \{ \{g_{j},c_{(j+1)}\} \mid j=1,\dots,\sqrt{n}-1\}.
\end{align*}

\begin{figure}
    \centering
    \includegraphics[width=1\linewidth]{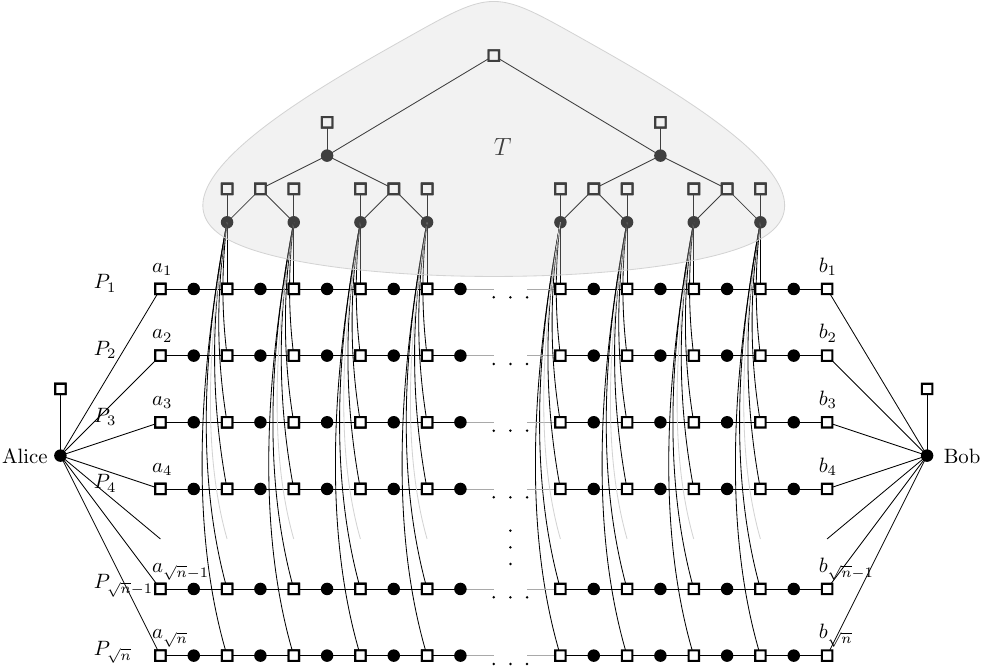}
    \caption{A visualization of the lower bound graph $SC_n$ where the boxes denote gifts and the circles denote children.}
    \label{fig:lb}
\end{figure}

Since we want to reduce the Santa Claus instance to a set-disjointness problem in the two-party model, we explain how Alice and Bob's input is incorporated in the construction. 
Let~$a_1a_2\dots a_{\sqrt{n}}$ and $b_1b_2\dots b_{\sqrt{n}}$ denote the $\sqrt{n}$-bit binary strings that Alice and Bob receive, respectively. 
Both Alice and Bob are child nodes and connected to $\sqrt{n}+1$ many gifts. One of those gifts is a private gift $p^A$ (resp. $p^B)$ that is used to make Alice (resp. Bob) happy, i.e., they have exclusive access to their gift. The remaining $\sqrt{n}$ gifts $g_i^A$ (resp. $g_i^B$) of value~$a_i$ (resp. $b_i$) with $i=1,\dots,\sqrt{n}$ are connected to Alice (resp. Bob) and the leftmost child $c_1^i$ (resp. rightmost child $c_{\sqrt{n}}^i$) of the path~$P_i$, see \cref{fig:lb}.

\paragraph{Hardness.}

In the following, we show how to use the graph $SC_n$ to obtain the hardness result of \cref{thm:SantaClausLowerBound}.

\begin{lemma}
\label{lem:set-disj-red}
    Any distributed algorithm that decides whether there exists a solution of value~$1$ in the graph $SC_n$ solves the set-disjointness problem for
    \begin{align*}
        A=\set{ i \mid a_i =0} \quad \text{and} \quad B=\set{i\mid b_i=0}.
    \end{align*}
\end{lemma}
\begin{proof}
    If $A \cap B=\emptyset$, then for every $i \in [\sqrt{n}]$ at least one of $a_i$ or $b_i$ is $1$. Without loss of generality assume that $a_i=1$, then we can assign the gift $g_i^A$ with value $a_i=1$ to $c_{1}^i$ and for the rest of the $i$th path assign $g_{j}^i$ to $c_{(j+1)}^i$, and thus every child $c_{j}^i$ gets assigned one gift of value $1$. The children in the binary tree, as well as Alice and Bob, get assigned their private gifts of value $1$, and thus also get assigned at least one gift each. 

    On the other hand, if $A\cap B \neq \emptyset$, there exists at least one index $i$ such that $a_i=b_i=0$. 
    Thus, the $i$-th path contains $\sqrt{n}$ children, but only $\sqrt{n}-1$ gifts of value $1$, and since each child is only adjacent to the gifts of that path, at least one child cannot get assigned any gift. 
    Thus, the value of this instance is $0$. 
\end{proof}

Due to \cite{DBLP:journals/tcs/Razborov92}, it is well known that the set-disjointness problem defined in \Cref{lem:set-disj-red} requires $\Omega(\sqrt{n})$ bits of communication, since the bit strings are of length $\sqrt{n}$.
Consequently, it follows from \Cref{lem:set-disj-red} that Alice and Bob need to communicate at least $\Omega(\sqrt{n})$ bits to solve the given Santa Claus instance.
To bound the number of bits Alice and Bob can communicate in $T$ rounds on the $SC_n$ graphs, we need the notion of moving cuts inspired by~\cite{doi:10.1137/S0097539700369740}. 

\begin{definition}
    Let $G=(V,E)$ be a graph with two distinguished nodes $a,b \in V$.
    For all nodes~$v,w \in V$, let $d_{1+\ell}(v,w)$ denote the shortest path distance from $v$ to $w$ where each edge has cost~$1+\ell(e)$.
    We define a \emph{moving cut} as an assignment $\ell: E \rightarrow \mathbb{Z}_{\geq 0}$ with capacity~$C(\ell)=\sum_{e\in E} \ell(e)$ and distance~$T(\ell) \le d_{1+\ell}(a,b)$.
\end{definition}

For the rest of this section, think of the two distinguished nodes $a$ and $b$ as Alice and Bob.
The following lemma is motivated by~\cite{peleg00}.

\begin{lemma}
\label{lem:bandwidth_latency}
    Let $G=(V,E)$ be a graph with two distinguished nodes $a,b \in V$ where each edge carries $k$ bits of information per round.
    Suppose $\ell: E \rightarrow \mathbb{Z}_{\geq 0}$ is a moving cut of capacity~$C(\ell)$ and distance $T(\ell)$. Then any distributed protocol that completes in at most~$T(\ell)-1$ rounds can communicate at most $C \cdot k$ bits from $a$ to $b$.
\end{lemma}

For a proof of \cref{lem:bandwidth_latency} refer to~\cite{PrinciplesOfDistributedComputing}.
In order to prove that any set-disjointness protocol on the graphs $SC_n$---and hence the Santa Claus problem---must run for at least $\widetilde{\Omega}(\sqrt{ n })$ rounds, we construct a moving cut with small capacity and large distance.

\begin{lemma}
\label{lem:moving_cut}
    For every $\gamma >0$ and $n \in \mathbb{N}$, there exists a moving cut $\ell_n$ in $SC_n$ with
    \begin{align*}
        C(\ell_n)=\mathcal{O}\left(\frac{1}{\gamma} \sqrt{n}\log n\right) \quad \text{and} \quad T(\ell_n)=\Omega\left(\frac{1}{\gamma}\sqrt{n}\right).
    \end{align*}
\end{lemma}
\begin{proof}
    Set $\ell(e)=0$ for all edges that are not part of the binary tree. Also set $\ell(e)=0$ for all edges between children and their private gifts in the binary tree. Let the height of a tree node be the distance to the leaves, i.e., the level of the leaves is $0$, and the level of their parents is $1$, and the level of the root is $\log(\sqrt{n})=\Oo(\log n)$. Assign to the parent edge of each node with height $h$ the value $\ell(e)=\lfloor \frac{2^h}{\gamma} \rfloor$. Since there are at most $\sqrt{n}/2^h$ nodes at height $h$, the total capacity of the moving cut $\ell$ is at most 
    \begin{align*}
        \sum_{h=0}^{\Oo(\log n)}\frac{\sqrt{n}}{2^h} \cdot \frac{2^h}{\gamma} = \sum_{h=0}^{\Oo(\log n)}\frac{\sqrt{n}}{\gamma} = \Oo\left(\frac{1}{\gamma}\sqrt{n}\log n\right).
    \end{align*}

    To check that any path from Alice to Bob is at least of length $\Omega\left(\frac{\sqrt{n}}{\gamma}\right)$ in the $1+\ell$ distance, we observe that if we shortcut $2k$ edges of a path $P_i$ via the tree, we have to use a path in $T$ that uses at least one node of height $\lceil \log k \rceil$. Thus, a path of length $2k$ (in one of the $P_i$'s) can be replaced with a path in the tree of $(1+\ell)$ distance at least
    \begin{equation*}
        \sum_{h=1}^{\lceil \log k \rceil} 2 \cdot\frac{2^{h-1}}{\gamma} = \sum_{h=1}^{\lceil \log k \rceil}\frac{2^h}{\gamma} \geq k/\gamma. \qedhere
    \end{equation*}
\end{proof}

Using the moving cut of \cref{lem:moving_cut}, we show that deciding if there exists a solution to the Santa Claus problem requires $\widetilde{\Omega}(\sqrt{n}/k)$ rounds.

\begin{lemma}
\label{cor:LB}
    Any distributed algorithm that uses $k$-bit messages per edge in each round to decide whether there is a solution to the Santa Claus problem of value at least $1$ in the $SC_n$ graph requires $\widetilde{\Omega}(\sqrt{n}/k)$ rounds.
\end{lemma}
\begin{proof}
    By \cref{lem:set-disj-red}, any distributed algorithm that decides whether there exists a solution to the Santa Claus problem of value at least $1$ solves set-disjointness for 
    \begin{align*}
        A=\set{ i \mid a_i =0} \quad \text{and} \quad B=\set{i\mid b_i=0}.
    \end{align*}
    From \cite{DBLP:journals/tcs/Razborov92}, it follows that any distributed protocol that solves this set-disjointness instance of size $\Theta(\sqrt{n})$ requires $\Omega(\sqrt{n})$ bits of communication between Alice and Bob. 
    For a sufficiently large constant $c>0$ and $\gamma = c\cdot k \cdot \log n$, \Cref{lem:moving_cut} states that there exists a moving cut $\ell_n$ with
    \begin{align*}
        C(\ell_n)=\mathcal{O}\left(\frac{\sqrt{n}}{c k}\right) \quad \text{and} \quad T\left(\ell_n\right)=\Omega\left(\frac{\sqrt{n}}{c k \log n}\right).
    \end{align*}
    By \Cref{lem:bandwidth_latency}, at most $C(\ell_n) \cdot k = \mathcal{O}(\sqrt{n}/c)$ bits can be communicated from Alice to Bob in~$T(\ell_n) - 1$ rounds. This is insufficient to solve set-disjointness, and hence decide the Santa Claus problem. Consequently, any algorithm requires at least $T= \widetilde{\Omega}(\sqrt{n}/k)$ rounds.
\end{proof}

Putting all the pieces together, we are now in the position to prove \cref{thm:SantaClausLowerBound}.

\begin{proof}[Proof of \Cref{thm:SantaClausLowerBound}]
    We argue that any \CONGEST multiplicative approximation algorithm that computes a (local) solution to the Santa Claus problem also solves the decision problem described in \Cref{cor:LB}. 
    If the solution size is $1$, then any (multiplicative) approximation algorithm will output a solution of size $>0$, and thus all children get assigned at least one gift of value $1$. 
    On the other hand, if the solution size is $0$, at least one child gets no gift of value $>0$. 
    Let $x_A$ (resp. $x_B$) be a bit that indicates whether all children Alice (resp. Bob) has access to get at least one gift. 
    In particular, $x_A=1$ if and only if all children that Alice has access to get at least one gift. 
    Now Alice and Bob can exchange $x_A$ and $x_B$ and output~$x_A \wedge x_B$, and thus decide whether it is a $0$ or $1$ instance; which implies the $\widetilde{\Omega}(\sqrt{n})$ lower bound by \Cref{cor:LB} with $k=\Oo(\log n)$.
\end{proof} 

\subsection{\boldmath\LOCAL Model}
\label{subsec:lower-bounds-local}

In the stronger \LOCAL model, a simpler instance suffices to show that we require at least diameter rounds to compute any (multiplicative) approximation to the integral Santa Claus problem.
\begin{lemma}
\label{lem:integral_local_diameter_LB}
    Any \LOCAL algorithm that computes any (multiplicative) approximation of the integral Santa Claus problem requires $\Omega(D)$ rounds in the worst case.
\end{lemma}
\begin{proof}
    Similar to the paths used in the construction of $SC_n$,  
    let $P$ be an alternating child-gift path that starts and ends with a child node and has $2n-1$ vertices in total.
    Thus, we have~$n$ child nodes and $n-1$ gift nodes and consider three instances:
    \begin{enumerate}
        \item $I_1$: $P$ remains unchanged with children on both ends,
        \item $I_2$: Add one gift to the leftmost child of the path,
        \item $I_3$: Add one gift to the rightmost child of the path.
    \end{enumerate}
    Observe that $I_1$ has solution size $0$, since there are more children than gifts, and $I_2$ and $I_3$ have a solution of size $1$, which can be obtained by assigning each gift to its left (resp. right) child neighbor; see \Cref{fig:path_assignment}.
    Thus, any approximation algorithm needs to learn in which direction the gifts have to be assigned, i.e., in particular, the middle gift of the path needs to know if it gets assigned to its left or right child neighbor. This can only be decided if this gift knows whether there is an additional gift at the left or right end of the path. Therefore, any \LOCAL approximation algorithm takes at least $D/2=\Omega(D)$ rounds.
\end{proof}

\begin{figure}
    \centering
    \includegraphics[width=1\linewidth]{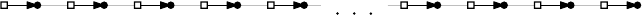}
    \caption{The gift assignments for instances $I_2$ and $I_3$ defined in \cref{lem:integral_local_diameter_LB} can be obtained by assigning each gift to its left or right child, respectively.}
    \label{fig:path_assignment}
\end{figure}

\subsection{Augmenting Paths and Fractional Solution}
\label{subsec:lower-bounds-augementing-paths}

Any alternating gift-child path that start at an unassigned gift and end with a child can be considered an augmenting path.
Both the lower bound family of graphs $SC_n$ and the single path construction of the previous sections show that the non-existence of short augmenting paths does not imply a good approximation.
In other words, the non-existence of short augmenting paths is not a certificate for a good approximation, thereby separating the Santa Claus problem from the Matching problem. 
We show that there exist two families of graphs where non-optimal solutions have only long augmenting paths. Notably, the non-optimal solutions do not provide any approximation guarantee.

\begin{corollary}\label{rem:no_augmenting_path}
    There exists a family of graphs, with diameter $\Oo(\log n)$, and a non-optimal solution such that any augmenting path has length $\Omega(\sqrt{n})$. In particular, the value of the non-optimal solution is $0$.
\end{corollary}

\begin{proof}
    Consider the graph class $SC_n$. 
    Assume that $a_i=1$ and $b_i=0$ for all $i=1,\dots,\sqrt{n}$. Consider the non-optimal solution where, in one path $P_i$, each gift $g_j^i$ is assigned to $c_j^i$. 
    Then~$c_{\sqrt{n}}^i$ has no assigned gift of value $1$, and thus the solution size is $0$. 
    Now the only augmenting path increasing the value of the current solution is $P_i$ itself, and thus has length~$\Omega(\sqrt{n})$.
\end{proof}

\begin{corollary}\label{rem:no_augmenting_path2}
    There exists a family of graphs and a non-optimal solution such that any augmenting path has length $\Omega(n)$. In particular, the value of the non-optimal solution is $0$.
\end{corollary}

\begin{proof}
    Consider the family of path graphs defined in \Cref{lem:integral_local_diameter_LB}. 
    Again, if all gifts are assigned to their left child in the instance $I_2$, the only augmenting path is the whole path, and thus has length $\Omega(n)$.
\end{proof}

Finally, we show that the fractional relaxation of the Santa Claus problem is already hard to solve in the distributed setting.
Via a reduction from the fractional perfect matching problem, we show the more general statement that any general fractional Mixed Packing-Covering LP solver requires at least diameter rounds. 

\begin{lemma}
    Any distributed algorithm that computes a feasible fractional solution or determines infeasibility for all Mixed Packing-Covering LPs takes at least $\Omega(D)$ rounds.
\end{lemma}

\begin{proof}
    We can model the perfect fractional matching problem for a given graph $G=(V,E)$ as a Mixed Packing-Covering LP:
    \begin{align*}
        & \quad  \sum_{e \in \delta(v)}x_e \leq 1 & \forall v \in V, \\
        & \quad \sum_{e \in \delta(v)}x_e \geq 1 & \forall v \in V, \\
        &\quad 0\leq x_e \leq 1 & \forall e \in E.
    \end{align*}
    To solve this problem, any distributed algorithm needs to distinguish between paths of odd length (LP is feasible) and paths of even length (LP is infeasible).
    This implies a lower bound of $\Omega(D)$ rounds.
\end{proof}

\newpage
\section{Solver for Mixed Packing-Covering LPs}
\label{sec:LP_solver}

In this section, we show how to solve mixed packing-covering LPs in the \CONGEST model. Our result strongly builds upon the parallel algorithm of~\cite{MRWZ16}, and our contribution is the implementation in the \CONGEST model. 
We start by formally defining two equivalent---up to scaling---formulations of mixed packing-covering LPs.

\begin{definition}[Mixed Packing-Covering LPs]
Given $\Pm \in \R_{\ge 0}^{n_p \times m},\Cm \in \R_{\ge 0}^{n_c \times m}, \pv \in \R_{\ge 0}^{n_p}$ and~$\cv \in \R_{\ge 0}^{n_c}$, the \emph{mixed packing-covering LP} $(\Pm, \Cm, \pv, \cv)$ is defined as
    \begin{align}
    \label{def:min-MPC}
        \min\{ \lambda \mid \Pm \x \leq \lambda \pv, \Cm \x \geq \cv, \x\geq 0 \}.
    \end{align}
    For any solution $\x$ to LP \eqref{def:min-MPC} of value $\lambda$, vector $\widetilde{\x} = \x/\lambda$ is a solution of value $\gamma = 1 /\lambda$ to
    \begin{align}
    \label{def:max-MPC}
        \max\{ \gamma \mid \Pm \x\leq  \pv, \Cm \x \geq \gamma \cv, \x\geq 0 \}
    \end{align}
    and vice versa.
\end{definition}

Notably, we have the following relation between the approximate solutions for both definitions of mixed packing-covering LPs.

\begin{restatable}{lemma}{equivalentFormulations}
    \label{rem:equivalent-formulations}
    Let $\lambda^*$ and $\gamma^*$ be the optimal values of \eqref{def:min-MPC} and \eqref{def:max-MPC}, respectively.
    Let $(\x,\lambda)$ be a solution to LP \eqref{def:min-MPC} with $\lambda^* \leq \lambda \leq (1+\eps)\lambda^*$.
    Then $(\widetilde{\x}, \gamma)$  is a solution to LP \eqref{def:max-MPC} with
    \begin{align*}
         \gamma^*/(1+ \eps) \leq \gamma \leq \gamma^* .
    \end{align*}
\end{restatable}

We defer the proof to Appendix \ref{sec:deferred_proofs}.
Throughout the rest of this section, we are interested in the formulation \eqref{def:max-MPC} where the covering constraints are violated instead of the packing constraints. 
We show the following main result.

\begin{restatable}{theorem}{ThmMixedLPSolver}\label{thm:solving_mixed_packing_covering}
    Let $\gamma^*$ be the optimal value of \eqref{def:max-MPC}.
    Given a mixed packing-covering LP~\eqref{def:max-MPC} with $n$ constraints and $m$ variables, we can compute a feasible solution $(\x,\gamma)$ such that 
    \begin{align*}
        (1 - \eps)\gamma^* \leq \gamma \leq \gamma^*
    \end{align*}
    in $\Otilde\left(\frac{D\log^2 n \log m}{\eps^3}\right)$ rounds of \CONGEST.
\end{restatable}

We start by formalizing the underlying communication network in \Cref{sec:underlying_comm}. In \Cref{sec:opt_to_feas}, we show how to solve the optimization problem via a reduction to the mixed packing-covering feasibility problem and a binary search framework.
In \Cref{sec:solve_feas}, we explain how to solve the feasibility problem in the \CONGEST model. 
Finally, we describe an implementation of the algorithm in the \CONGEST model with small messages of size~$\Oo(\log n)$ in \Cref{subsec:impl_cong}---therefore, we do not address message sizes in previous sections.

\subsection{Underlying Communication Networks}
\label{sec:underlying_comm}
Formally, we consider a bipartite graph $G=(V,E)$ with $V=C \cup X$ where $|C|=n_c+n_p$ contains a node for each covering and each packing constraint and $|X|=m$ a node for each variable. 
Each variable node $i\in[m]$ is connected to a packing constraint node $j_p \in [n_p]$ or a covering constraint node $j_c \in [n_c]$ if the corresponding entry of $\Cm$ or $\Pm$ is non-zero.
More formally, we have $E=\{\set{i,j_p} \mid \Pm_{ij_p} \neq 0\} \cup \{ \set{i,j_c} \mid\Cm_{ij_c}\neq 0\}$. 
Additionally, each packing constraint node $j_p$ and covering constraint node~$j_c$ has access to the $j_p$-th row $\Pm_{j_p}$ and the~$j_c$-th row~$\Cm_{j_c}$, respectively. 
Slightly abusing notation, we omit the additional index and refer to covering and packing constraints by just~$j$.
Furthermore, we assume that every node implicitly knows~$n_p,n_c$, and~$m$---otherwise this can be computed in $\Oo(D)$ rounds.

In \Cref{sec:Santa_Claus}, our communication network is a bipartite child-gift graph $G=(\mathcal{C} \cup \mathcal{G}, E)$ where there is an edge $\{c,g\} \in E$ present if the child $c$ desires $g$. 
In the following lemma, we show that these two communication models are equivalent up to a constant round overhead, i.e., we can invoke the algorithm of \Cref{thm:solving_mixed_packing_covering}  using the communication network of \Cref{sec:Santa_Claus}.

\begin{restatable}{lemma}{EquivalentModels}
    Given a weighted bipartite graph $G=(\mathcal{C} \cup \mathcal{G}, E)$ and vertex weights $v_g \in \mathbb{R}_{\geq 0}$ for all $g \in \mathcal{G}$, we can simulate the algorithm of \Cref{thm:solving_mixed_packing_covering} with input $($\LP$,\eps)$, on~$G$ such that any round of the algorithm takes $\Oo(1)$ rounds.
\end{restatable}

We defer the proof to Appendix~\ref{sec:deferred_proofs}.

\subsection{Reducing Optimization to Feasibility}
\label{sec:opt_to_feas}

Similarly to other sequential and parallel approaches for solving mixed packing-covering LPs, we reduce the optimization problem \ref{def:min-MPC} (or equivalently \ref{def:max-MPC}) to the $(1+\eps)$-feasibility problem via a standard binary search and scaling argument.
For the sake of completeness, we include the reduction that is done explicitly in \cite{Young2001} and stated implicitly in \cite{MRWZ16}.

\paragraph{\boldmath$(1+\eps)$-Feasibility.} Similar to existing literature, we define the $(1+\eps)$-feasibility problem with respect to LP~\eqref{def:min-MPC}.

\begin{definition}[$(1+\eps)$-Feasibility Problem]
\label{def:feasibility_prob}
    Given $\Pm \in \R_{\ge 0}^{n_p \times m},\Cm \in \R_{\ge 0}^{n_c \times m}, \pv = \ones$ and~$\cv = \ones$, the mixed packing-covering $(1+\eps)$-feasibility problem is to either find $\x\geq 0$ such that 
    \begin{align}
    \label{cond:feas}
        0<\max_j(\Pm \x)_j\leq(1+\eps)\min_j(\Cm \x)_j
    \end{align}
    or conclude that the following LP is  infeasible
    \begin{align}
        \Cm \x& \geq \ones \nonumber \\
        \Pm \x& \leq (1-10 \eps) \ones \label{LP:feas_prob}\\
        \x & \geq 0. \nonumber
    \end{align}
\end{definition}

The feasibility problem and the optimization problem are closely related. 
In the sequential setting, it is common to combine the feasibility problem with a binary search framework to obtain an approximation for the optimization problem. 
In the distributed setting, this technique directly implies a running time of $\Omega(D)$, as we have to globally communicate the next value for which we want to test feasibility. 
However, since many global problems---including the Santa Claus problem---have an $\Omega(D)$ lower bound, it is possible to use this technique. 
In \cref{sec:solve_feas}, we show how to solve the $(1+\eps)$-feasibility problem in the \CONGEST model.
More formally, we show the following lemma.

\begin{lemma}
\label{thm:solve_feasibility}
    There exists a deterministic \CONGEST algorithm for the mixed packing-covering $(1+\eps)$-feasibility problem with round complexity 
    \begin{equation}
        \label{eq:solve_feasibility}
        \Oo\left(D \cdot \frac{1000 \ln^2 n \ln\left(\frac{m}{\eps}\right)}{\eps^3}\right).
    \end{equation}
\end{lemma}

\paragraph{Reduction.} Any mixed packing-covering LP $(\Pm,\Cm,\cv,\pv)$ can be rescaled to solve the problem 
\begin{align*}
    \min \set{ \lambda \mid \widetilde{\Pm} \x \leq \lambda \cdot 1, \widetilde{\Cm} \x \geq 1, \x\geq 0},
\end{align*}
where each entry of the $j$th row of $\Pm$ and $\Cm$ is divided by $\pv_j$ and $\cv_j$, respectively.
We show that this formulation can be solved approximately, i.e., we can find $(\lambda,\x)$ such that 
\begin{align*}
    \lambda^* \leq \lambda \leq (1+\eps)\lambda^*.
\end{align*}
Using binary search, we choose different values for $\lambda'$ and $\eps'$ and find a sequence of solutions to the feasibility problem invoked with~$\Pm/\lambda',\Cm$ and $\eps'$.
Formally, \cite{Young2001} showed the following.

\begin{lemma}[\cite{Young2001}]
    \label{lem:young-approx}
    The approximate optimization problem reduces to a sequence of approximate feasibility subproblems: 
    \begin{itemize}
        \item $\Oo(\log \log m)$ subproblems with $\eps'=1/2$ and 
        \item $\Oo\left(\log (1/\eps)\right)$ subproblems where the $i$th--to--last subproblem has $\eps'=\Omega\left(\eps\left(\frac{4}{3}\right)^i\right)$.
    \end{itemize}
\end{lemma}

For more details on the choice of subproblems, refer to \cite{Young2001}.
Importantly, the running time of \cref{lem:young-approx} is dominated by solving the feasibility problem $\Oo(\log \log m)$ many times for~$\eps'=1/2$ and solving one feasibility problem with $\eps'=\eps$.
Due to \Cref{thm:solve_feasibility}, the number of rounds necessary to solve each of these subproblems is bounded by \cref{eq:solve_feasibility}.
In order to execute a binary search in \CONGEST, we use \Cref{lm:aggregation} to broadcast both the outcome and the next value for $\lambda'$ to all the nodes in $\Oo(D)$ rounds.
With the broadcast overhead for binary search, we solve any given mixed packing-covering LP in a total running time of 
\begin{align*}
    &\Oo \Bigg( \log \log m \cdot  D \log^2 n \log m + D \cdot  \frac{\log^2 n \log \left(\frac{m}{\eps}\right)}{\eps^3}
    \Bigg).
\end{align*}

\subsection{Solving the Feasibility Problem}
\label{sec:solve_feas}

In this section, we provide the missing proof of \cref{thm:solve_feasibility} by implementing the \emph{parallel} algorithm of \cite{MRWZ16} in the \CONGEST model, for the pseudocode refer to \cref{alg:mixed_covering/packing}.
First, we give a short description of \cref{alg:mixed_covering/packing} in \cref{sec:algo}. Second, for the sake of completeness, we include a brief overview of the correctness proof of \cite{MRWZ16} in \Cref{sec:correctness}.
Finally, we bound the round complexity of \cref{alg:mixed_covering/packing} in \cref{subsubsec:round-complexity}.

\subsubsection{Algorithm for the Approximate Feasibility Problem}
\label{sec:algo}
At a high level, \Cref{alg:mixed_covering/packing} starts with a small nonnegative solution $\x$ and iteratively increases selected entries of $\x$. In each iteration, the algorithm compares the progress made towards satisfying the covering constraints with the corresponding increase in the packing constraints. To this end, constraints are weighted according to how close they are to becoming critical, so that highly loaded packing constraints and poorly satisfied covering constraints have a stronger influence on the update. Based on these weights, the algorithm determines which variables can be increased while making sufficiently more progress on the covering side than on the packing side. The process continues until the constraints have reached a sufficiently large scale, at which point the current solution is rescaled and returned. If no variable admits such an update, the algorithm concludes that the corresponding feasibility instance is infeasible. We describe the quantities used to implement these steps in more detail in the following.

We start by setting the initial solution in Line 3 of \cref{alg:mixed_covering/packing} to
\begin{equation*}
    \x_i=\frac{1}{m ||\Pm_{:,i}||_{\infty}}, 
\end{equation*}
for all $i\in [m]$.\footnote{For better readability, we omit the iteration counter throughout this description.}
Then the values for $x$ are increased until it holds that
\begin{equation*}
    \max\left\{\max_j\set{\Pm \x}_j, \min_j \set{\Cm \x}_j \right\} \geq K = \frac{10 \ln n}{\eps}. 
\end{equation*}
The two components of the stopping condition can be found in Line 7 for the packing constraints and in Line 13 for the covering constraints. 
We implicitly encode $\max_j \set{\Pm \x}_j \geq K$ by stopping as soon as either \emph{one} packing constraint fulfills $\Pm_j \x \geq K$ in Line 11. To encode $\min_j \set{\Cm x}_j \geq K$ we stop if \emph{all} covering constraints fulfill~$\Cm_j \x \geq K$ in Line 17. 
In contrast to~\cite{MRWZ16}, it is not necessary to explicitly delete all covering constraints of size at least $K$.
Instead, we simply do not communicate the values $z_j$ and $b_{ij}$ for these constraints as they do not contribute anything to the computation of $b_i$.

In the for loops for the packing and covering constraints, we compute $y_j$ and $z_j$, respectively. Note that these values are exponentials of one packing or covering constraint and thus the vectors $\set{y_i}_{i=1}^{n_p}$ and $\set{z_i}_{i=1}^{n_c}$ are exponentials of the packing and covering constraints, respectively. These exponentials are used to weigh the update step for the individual variables.
Similarly, we obtain $a_{ij}$ and $b_{ij}$ to compute $a_i$ and $b_i$ for each variable $i$---note that $\set{a_i}_{i=1}^m$ and~$\set{b_i}_{i=1}^m$ can be interpreted as gradients of~$\lmax(\Pm \x)$ and $\lmin(\Cm \x)$ that are used to update~$\x$. In particular, in Line 26, we multiplicatively update any entry $\x_i$ by a factor dependent on $\frac{a_i}{b_i}$ if~$a_i \leq (1-\frac{\eps}{50})b_i$.

Note that we encode the infeasibility condition $\set{i \mid a_i \leq (1-\frac{\eps}{50})b_i}=\emptyset$ by checking each variable in Line 23.
If all of them broadcast \emph{infeasible}, then each variable returns \emph{infeasible}. 

\begin{algorithm}
    \caption{Distributed Algorithm for the Approximate Feasibility Problem}
    \label{alg:mixed_covering/packing}
    \SetKwInOut{Input}{Input}
    \SetKwInOut{Output}{Output}
    
 \newcommand\mycommfont[1]{\textcolor{gray}{\itshape #1}}
    \SetCommentSty{mycommfont}
    \SetKwComment{tcp}{}{}
    \DontPrintSemicolon

        \Input{$\Pm,\Cm,\eps$}
        \Output{ $\textit{infeasible or } \x_i \geq 0\textit{ at each node } i \textit{ s.t.} \max_j(\Pm \x)_j \leq (1+\eps) \min_j(\Cm \x)_j $}

        Let $R=\Oo\left(\frac{1000 \ln^2 n \ln(\frac{m}{\eps})}{\eps^3}\right),  K=\frac{10 \ln n}{\eps}$ and $t=0$. 

        \For{each variable $i$}{
        
        $\x_i^{(0)}=\frac{1}{m ||\Pm_{:,i}||_{\infty}}$ \tcp*[f]{---Each variable $j$ can send $p_{ji}$ to the variable $i$ in one round}
        
        Send $\x_i^{(0)}$ to all constraints $j \in N(i)$ 
         
        }
        \While{not Stop and $t \leq R$}{
        
        \For(\tcp*[f]{---$j$ has access to $\{p_{ji}\}_{i=0}^{n}$ and $\x_i$ for all $i \in N(j)$}){each packing constraint $j$ }{ 
        \If{$(\Pm_j\x^{(t)}) < K$}{
        
        $y_j^{(t)}=\exp(P_j\x^{(t)})=\exp(\sum_{i \in N(j)}p_{ji}\x_i^{(t)})$
        
        $a_{ij}^{(t)}=p_{ji}y_j^{(t)}$ for all $i \in N(j)$
        
        To each $i \in N(j)$ send $y_j^{(t)}$  and $a_{ij}^{(t)}$ 
        
        \lElse{ Broadcast Stop$_p^{(t)}$ \tcp*[f]{---If one packing constraint broadcasts Stop$_p$, we abort}} 
        }
        }
        \For{each covering constraint $j$}{
        \If{$(\Cm_j\x^{(t)})<K$}{
        $z_j^{(t)} = \exp(-\Cm_j\x^{(t)}) = \exp(-\sum_{i\in N(j)}c_{ji}\x_i^{(t)})$,
        
        $b_{ij}^{(t)} = c_{ji}z_j^{(t)}$
        
        To each $i \in N(j)$ send $z_j^{(t)}$ and $b_{ij}^{(t)}$
        
        \lElse{ Broadcast Stop$_c^{(t)}$ \tcp*[f]{---If all covering constraints broadcast Stop$_c$, we abort}}  
        }
        }
        \For{each variable $i$}{
        
        $y^{(t)} = \sum_{j=1}^{n_p} y_j^{(t)}$ \tcp*[f]{---$y^{(t)}$ and $z^{(t)}$ can be computed in $\Oo(D)$ rounds, see \Cref{lem:inner_prod}} 
        
        $z^{(t)}=\sum_{j=1}^{n_c} z_j^{(t)}$ \tcp*[f]{---Only consider covering constraints that sent values $z_j,b_{ij}$}
        
        $a_i^{(t)}=\frac{\sum_{j\in N(i) }a_{ij}^{(t)}}{y^{(t)}}$
        
        $b_i^{(t)}=\frac{\sum_{j\in N(i) }b_{ij}^{(t)}}{z^{(t)}}$
        
        \If(\tcp*[f]{---If all variables broadcast "infeasible", return "infeasible"})
        {$a_i^{(t)} > (1-\frac{\eps}{50})b_i^{(t)}$ }
        {
        
        $\Delta_i^{(t)} = 0$ and Broadcast infeasible$^{(t)}$ 
        
        \Else{ $\Delta_i^{(t)} = \frac{1}{2}(1-\frac{a_i^{(t)}}{b_i^{(t)}}) \in [\eps/100,\frac{1}{2}]$}
        }
        
        $\x_i^{(t+1)} \leftarrow \x_i^{(t)}(1+ \frac{1}{K} \Delta_i^{(t)})$  
        
        $t\leftarrow t+1$
        
        To each constraint $j \in N(i)$ send $\x_i^{(t+1)}$ and $t$ 
        } }
        \lIf{$t > R$}{
        \Return "infeasible"
        }
        
        Let $T$ be the first iteration, where one packing constraint broadcasted Stop$_p^{(T)}$ or all covering constraints broadcasted Stop$_c^{(T)}$.
        
        \lFor{each variable $i$}{
        \Return $\x_i=\frac{\x_i^{(T)}}{K}$ \label{line:return}
        }  
\end{algorithm}

\subsubsection{Correctness }

\label{sec:correctness}

For the sake of completeness, we include a brief overview of the correctness proof of \cite{MRWZ16} that also shows correctness for our distributed implementation.
We begin by showing that \Cref{alg:mixed_covering/packing} terminates and outputs the correct answer. 

\begin{lemma}[\cite{MRWZ16}, Lemma 2] 
    \label{lem:infeasible}
    If the problem instance is feasible, then we have
    \begin{align*}
        \forall \x\geq 0, B=\{i \mid a_i \leq (1-\frac{\eps}{50})b_i\} \neq \emptyset.
    \end{align*}
\end{lemma}

\Cref{lem:infeasible} shows that the fact that all variables broadcast infeasible certifies that the algorithm correctly terminates when the input instance is infeasible.
If it never outputs infeasible, then at least one variable $\x_i$ gets increased by factor of 
\begin{equation*}    
    (1+\frac{1}{K}\Delta_i) \geq (1+\frac{1}{K}\frac{\eps}{100})=(1+\frac{\eps^2}{1000 \ln n}) 
\end{equation*}
in each iteration of the while loop.
Consequently, the algorithm must reach one of the termination conditions after a finite number of iterations. 

In the remainder of this section, we show that the output $\x=\{\x_i\}_{i=1}^n$ satisfies 
\begin{equation*}    
    0<\max_j(\Pm \x)_j \leq (1+\eps) \min_j(\Cm \x)_j.
\end{equation*}
For the analysis, we consider the potential function $f(\x)=\lmax (\Pm \x)-\lmin(\Cm \x)$ where the soft-max $\lmax(\Pm \x)$ and soft-min $\lmin(\Cm \x)$ are defined as follows
\begin{align*}
    \lmax(\Pm \x)= \ln \Big(\sum\nolimits_{j}\exp (\Pm \x)_j \Big) \quad \textit{ and } \quad \lmin(\Cm \x )= - \ln\Big( \sum\nolimits_{j} \exp(-(\Cm \x)_j) \Big). 
\end{align*}

Then the following lemma bounds the difference of $\lmax(\Pm\x)$ and $\lmin(\Cm\x)$ by $2\ln n$.

\begin{lemma}[\cite{MRWZ16}, Lemma 4]
    \label{lem:lemma4}
    Given $\max_j(\Pm \x^{(t)})_j<\frac{10 \ln n}{\eps}$ and $\min_j(\Cm \x^{(t)})_j<\frac{10 \ln n}{\eps}$, we always have $f(\x^{(t)})\leq2 \ln n$ during the execution of \Cref{alg:mixed_covering/packing}.
\end{lemma}

It follows from \cref{lem:lemma4} that these functions approximate $\max_j(\Pm\x)_j$ and $\min_j(\Cm\x)_j$:
\begin{align*}
     \max_j(\Pm\x)_j & \leq \lmax(\Pm\x) \leq \max_j(\Pm\x)_j +\ln n, \\
     \min_j(\Cm\x)_j - \ln n &\leq  \lmin(\Cm\x) \leq  \min_j(\Cm\x)_j.
\end{align*}
This yields the desired approximation once the terms are large at the time of termination. 
\begin{lemma}[\cite{MRWZ16}, Lemma 5]
    If the algorithm terminates in Line \ref{line:return}, then $\x=\{\x_i\}_{i=1}^n$ fulfills $\x\geq 0$ with $0<\max_j(\Pm \x)_j\leq (1+\eps)\min_j(\Cm \x)_j.$
\end{lemma}

Until now, we showed that \cref{alg:mixed_covering/packing} either terminates and outputs $\{\x \}_{i=1}^n$ that satisfies~(\ref{cond:feas}) or terminates early and correctly certifies the infeasibility of LP (\ref{LP:feas_prob}). 

\subsubsection{Round Complexity}
\label{subsubsec:round-complexity}

To finish the proof of \Cref{thm:solve_feasibility}, it remains to show that if the LP (\ref{LP:feas_prob}) is feasible, then \Cref{alg:mixed_covering/packing} terminates after 
\begin{equation*}
    \Oo\left(D\cdot \frac{1000 \ln^2 n \ln(\frac{m}{\eps})}{\eps^3}\right)
\end{equation*}
rounds.
In case it takes more rounds, we simply terminate and output that the LP (\ref{LP:feas_prob}) is infeasible.
In a nutshell, we show that every iteration of the parallel algorithm of \cite{MRWZ16} runs in~$\Oo(D)$ rounds in the \CONGEST model. Then their bound on the number of iterations 
\begin{align}
\label{eq:number_of_it}
    \Oo \left(\frac{\log n}{\eps} \cdot \left(\frac{\ln n \ln (m/\eps)}{\eps}+\frac{\ln n \ln(m/\eps)}{\eps^2}\right)\right)
    = \Oo\!\left(\frac{\log^2 n \log (m/\eps)}{\eps^3}\right)
\end{align}
yields the overall desired running time.
We start by showing that the sums $y^{(t)}$ and $z^{(t)}$ can be computed efficiently.

\begin{lemma}
    \label{lem:inner_prod}
    Let $0 \leq t\leq T$. The sums $y^{(t)}$ and $z^{(t)}$ can be computed in $\Oo(D)$ rounds.
\end{lemma}

\begin{proof}
Compute a BFS tree rooted at an arbitrary node $v\in V$, e.g., the largest ID node in $\Oo(D)$ rounds. Now starting from the leaves, each node $i$ sends the sum of the value it received from its children and $y_i^{(t)}$ ($z_i^{(t)}$ respectively) to its parent. If the node is a constraint, it just sends the information forward towards $v$ and does not change it. Then after $\Oo(D)$ rounds the two values $y^{(t)}$ and $z^{(t)}$ are stored in $v$. These can now be broadcasted to all the nodes using the computed BFS tree.
\end{proof}

We conclude by bounding the round complexity of \Cref{alg:mixed_covering/packing}.

\begin{lemma}
\label{lem:rounds_per_it}
    Each iteration of the while loop of \Cref{alg:mixed_covering/packing} runs in $\Oo(D)$ rounds.
\end{lemma}
\begin{proof}
    We will consider an arbitrary iteration $0 \leq t \leq T$.
    First observe that at the beginning iteration $t$ before the while loop in line 6 starts, each packing constraint $j_p$ and each covering constraint $j_c$ has access to $\x_i^{(t)}$ for all $i\in N(j_p)$ and $i \in N(j_c)$, respectively, which holds due to line 5 and line 29. Thus the computations of $y_j^{(t)}, a_{ij}^{(t)},z_j^{(t)}$ and $b_{ij}^{(t)}$ can be executed locally in the node corresponding to the respective constraint. This finishes the computations executed at constraints.
    
    Since $y^{(t)}$ and $z^{(t)}$ can be seen as inner products \Cref{lem:inner_prod} shows that they can be computed in $\Oo(D)$ rounds. Due to line 11 and line 17, each variable $i$ has access to $a_{ij}^{(t)},y_j^{(t)},b_{ij}^{(t)}$ and $z_j^{(t)}$ for all $j\in N(i)$. This allows to execute lines 22 to 28 locally. 
    
    Finally, observe that broadcasting the stop tokens Stop$_p^{(t)}$ and Stop$_c^{(t)}$ can be done in $\Oo(D)$ rounds and does not yield congestion, since we can represent each token by 2 bits; one representing whether the corresponding constraint wants to stop or not and the other one indicating if the sending constraint is a packing or covering constraint. A similar behavior happens for the broadcast of "infeasible$^{(t)}$" in line 25, which can also be done in $\Oo(D)$ rounds. 
\end{proof}

We are now in the position to prove our feasibility result.

\begin{proof}[Proof of \cref{thm:solve_feasibility}]
    Since the correctness follows from \cref{sec:correctness}, it remains to show the desired round complexity. 
    Combining the bound on the number of iterations in \Cref{eq:number_of_it} with \Cref{lem:rounds_per_it} shows that the algorithm terminates after 
    \begin{equation*}
        \Oo\left(D\cdot\frac{\log^2n \log \left(m/\eps\right)}{\eps^3}\right)
    \end{equation*}
    rounds, if the LP is feasible. 
    Otherwise, there are two cases in which the algorithm terminates and outputs infeasible: either based on the infeasibility condition or after it reached
    \begin{equation*}
        \Oo\left(D\cdot\frac{1000\log^2n \log \left(m/\eps\right)}{\eps^3}\right) 
    \end{equation*} 
    rounds.
    Finally, it might be the case that we have to wait an additional $\Oo(D)$ rounds to gather all the stop tokens Stop$_c^{(t)}$ and Stop$_p^{(t)}$, so we can return $\frac{\x_i^{(T)}}{K}$.
    Thus, the algorithm stops at the latest after the following number of rounds
    \begin{equation*}
         \Oo\left(D\cdot\frac{1000\log^2n \log \left(m/\eps\right)}{\eps^3}+D\right)= \Oo\left(D\cdot\frac{1000\log^2n \log \left(m/\eps\right)}{\eps^3}\right)  \qedhere
    \end{equation*}
\end{proof}

Together with \Cref{rem:equivalent-formulations}, this almost completes the proof of \Cref{thm:solving_mixed_packing_covering}. All that is left to show is how to implement \Cref{alg:mixed_covering/packing} with small message size, and still high (enough) precision, which we show in \Cref{subsec:impl_cong}.

\subsection{Finite precision and \boldmath\CONGEST Implementation}
\label{subsec:impl_cong}
In this section, we show that \Cref{alg:mixed_covering/packing} can be implemented with small messages of size~$\Oo(\log n)$ in the \CONGEST model.
Notably, for $\eps < 1/n^{1.5}=1/\poly n$, the running time of \Cref{thm:solving_mixed_packing_covering} exceeds $\Oo(n^2)$.
Thus, we can trivially match the running time of \Cref{thm:solving_mixed_packing_covering} by gathering the whole graph in one node and locally solving the problem. As a consequence, we only consider the case where $\eps > 1/\poly n$ in the following lemma.

\begin{restatable}{lemma}{finiteprecision}
    \label{rem:finite_precision_congest}
    For $\eps > 1/\poly n$, \Cref{alg:mixed_covering/packing} can be implemented in the \CONGEST model with rational messages of $\Oo(\log n)$ bits and rational output values $x_i$.

\end{restatable}

We defer the proof of \Cref{rem:finite_precision_congest} to \Cref{sec:deferred_proofs}. As a last step, we conclude this section by proving \Cref{thm:solving_mixed_packing_covering}.

\begin{proof}[Proof of \Cref{thm:solving_mixed_packing_covering}]
    Correctness and runtime follow from the described reduction in \Cref{sec:opt_to_feas} combined with \Cref{thm:solve_feasibility}. \Cref{rem:finite_precision_congest} shows that we can implement the algorithm in the \CONGEST model, finishing the proof.
\end{proof}

\newpage
\printbibliography[heading=bibintoc]

\appendix

\section{A Counterexample for Naive Sparsification}\label{sec:example}

In this section we provide a concrete fractional solution to Santa Claus which cannot be rounded to an approximately optimal integral solution if we restrict ourselves to only consider the subset of non-zero fractional edges for rounding.

\begin{example}\label{example:sparsification}
    Let $k$ be an arbitrary integer. 
    We give an example of a graph, together with an optimal fractional solution, where any $k$-approximation of the optimal integral solution requires edges that are not part of the fractional solution. 
    We consider an example with $k$ children, $k-1$ big gifts of value $T$ and $T$ many small gifts of value $1$, where $T$ is a multiple of $k$. 
    Consider the graph depicted in \Cref{fig:sparsification_example}.
    In this example, each big gift has value $T$ and each small gift has value $1$.  
    Every child desires every small gift in this example. 

    An optimal integral solution in this example has value $T$. 
    It consists of assigning $k-1$ children a big gift and assigning one child all the small gifts. 

    We give a fractional solution of value $T$ as follows. 
    The fractional solution for the big gifts is depicted by weights on the edges. Moreover, the fractional solution assigns $T/k $ small gifts exclusively to each child (with fractionality 1), depicted by the blue sets.     
    In this solution, each child receives $(k-1)/k \cdot T$ of fractional value by the big gifts and $T/k$ value by small gifts, for a total of $T$. 

    Most edges between children and small gifts now have fractionality $0$.  
    That is, the value that any child in this tree can get from small gifts alone is capped by $T/k$, unless edges with fractionality 0 can be considered for the integral solution.
    Since there has to be a child that does not receive a big gift in any integer allocation, we can only achieve a value of $T / k$, if we ignore edges with fractionality 0 in this particular fractional assignment.
\end{example}

\begin{figure}
    \centering
    \includegraphics[width=0.5\linewidth]{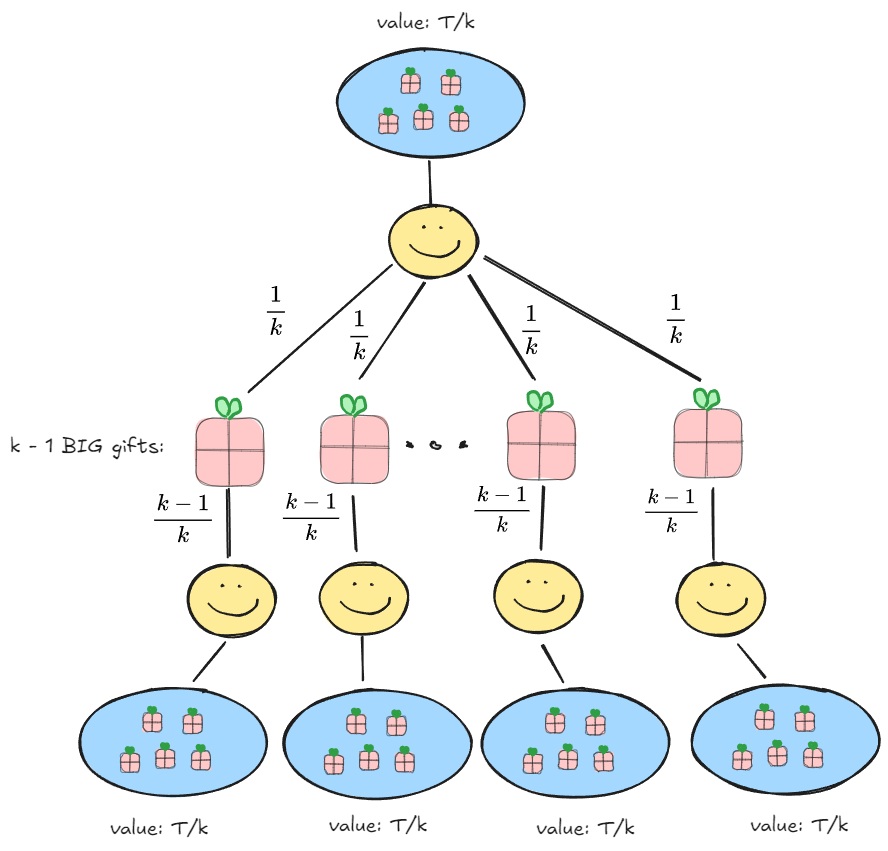}
    \caption{A fractional assignment of gifts in a graph with $k$ children, $k-1$ big gifts of value $T$ and $T$ small gifts of value $1$.}
    \label{fig:sparsification_example}
\end{figure}

\section{\boldmath \CONGEST Primitives}
\label{subsec:preliminaries}

In this section, we provide several \CONGEST primitives that we use throughout our algorithms.

\paragraph{Broadcast and Convergecast.}
In the \CONGEST model, there are standard Broadcast and Convergecast procedures, also known as ``pipelining''~\cite{peleg00}. Here information is either broadcasted \emph{from} one node or gathered \emph{in} one. 
We include a proof here for completeness. 

\begin{lemma}\label{lm:aggregation}
    There is a $\CONGEST$ algorithm that, given a graph $G=(V,E)$ with information of $t$ bits divided over the nodes, aggregates this information at a single node $v\in V$ in $\Oo(t/\log n+D)$ rounds. Conversely, in $\Oo(t/\log n+D)$ rounds we can make $t$ bits known by a single node $v\in V$, known to all nodes in the network.
\end{lemma}

\begin{proof}
    Create a BFS-tree rooted at $v\in V$ in $\Oo(D)$ rounds. Upcast the information along this tree. For a single bit, this takes at most $\Oo(D)$ rounds. Since each message can contain at most $\Oo(\log n)$ bits, the congestion can be at most $t/\log n$. Using the packet routing scheme from \cite{PacketRouting} this takes at most $\Oo(t/\log n+D)$ rounds. The converse result is obtained by downcasting the information towards the leaves.
\end{proof}

\paragraph{Rake and Compress.}

For our algorithm, we often need to run subroutines on subgraphs or topological minors of our communication network. These graphs might contain long paths and thus have a much higher diameter compared to our original communication network. 
To overcome this issue, we show how to circumvent the dependency on the diameter of the input graph by shortcutting via the additional edges in the communication graph. 
Similar techniques are for example used in~\cite{ForsterGLPSY21,RozhonGHZL22}. 
We will use two simple operations inspired by Miller \& Reif \cite{MillerR89}.
\begin{itemize}
    \item \textit{Rake:} Remove all vertices of degree at most one and all incident edges.
    \item \textit{Compress:} For each path $P = (v_0,v_1,\dots,v_k)$ such that $\mathrm{deg}(v_i) = 2$ for all ${i = 1,\dots, k - 1}$, we replace $P$ by the single edge $\{v_0,v_k\}$.
\end{itemize}

\begin{lemma}\label{lem:compress}
    Let $G=(V,E)$ denote an $n$-node \CONGEST communication network of diameter $D$ and let $G'$ be a topological minor of $G$.
    Assume that every vertex $v \in V$ knows its neighbors in $G'$.
    Then, there is a deterministic algorithm that computes the graph $G'' = \textit{Compress}(G')$ in $\Otilde(\sqrt n+D)$ rounds of \CONGEST on $G$. 
    Computing the graph $G''$ means that for every edge $uv \in E$ both $u$ and $v$ know whether $uv \in E[G'']$.
    In addition, if the edge $uv$ got removed by a \textit{Compress} operation, then both $u$ and $v$ know the edge $u'v' \in E[G'']$ that replaced $uv$ as well as the immediate neighbor on the path towards $u'$ and $v'$, respectively.
    Furthermore, from now on, one round of \CONGEST on $G''$ can be simulated in $\widetilde{\mathcal{O}}(\sqrt{ n } + D)$ rounds of \CONGEST on $G$.
\end{lemma}

\begin{proof}
    Let $\mathcal{P}$ be the set of all maximal paths $P = (v_0,v_1,\dots,v_k)$ in $G'$ such that $v_0 \neq v_k, k \geq 2$ and $\mathrm{deg}(v_i) = 2$ for all $i = 1, \dots, k - 1$.
    Note that $\mathcal{P}$ is a set of edge-disjoint paths and that two paths in $\mathcal{P}$ may only intersect at the endpoints of both paths.
    For each path $P \in \mathcal{P}$, let $P'$ denote the corresponding path in $G$ obtained by replacing each edge in $G'$ that does not appear in $G$ with its corresponding subdivision, which is a path in the original graph $G$.
    Using this terminology, we call any path $P \in \mathcal{P}$ \emph{short} if the length of $P'$ in $G$ is at most $\sqrt{ n }$, and \emph{long} otherwise.
    For any short path $P$, each vertex on $P$ can learn both endpoints of $P$ in $\mathcal{O}(\sqrt{ n })$ rounds of \CONGEST on $G$.
    Furthermore, each vertex can also learn which incident edges belong to a long path in the same runtime.
    We denote the graph induced by the edges of all long paths as $G_\mathrm{long} \subseteq G$.
    Since traversing these long paths all the way to the end would be prohibitively expensive, we need to find a way to break them up into segments of manageable size.
    For that purpose we compute a $(\sqrt n, \mathcal{O}(\sqrt n \cdot \log n))$-ruling set $R$ on $G_\mathrm{long}$ in $\mathcal{O}(\sqrt{n} \cdot \log n)$ rounds on $G$ (implicit in \cite{GoldbergPS88}, explicit in \cite{HenzingerKN21}). 
    Next, we note that by a simple vertex counting argument there can only be $\mathcal{O}(\sqrt n)$ long paths in $\mathcal{P}$.
    Furthermore, the size of $R$ is also bounded by $\mathcal{O}(\sqrt{ n })$, since each long path of length $k$ contains at most $\lceil k / \sqrt{ n } \rceil $ vertices from the ruling set $R$.
    Let $V_\mathrm{long}$ denote the set of endpoints of all long paths in $\mathcal{P}$.
    Now we consider the graph $G_\mathrm{long}' = (V_\mathrm{long} \cup R, E)$, where two vertices $v$ and $w$ are connected by an edge in $G_\mathrm{long}'$ if and only if they are connected in $G_\mathrm{long}$ and there is a path between $v$ and $w$ that does not contain any other vertex from $V_\mathrm{long} \cup R$.
    The graph $G_\mathrm{long}'$ consists of at most $\Otilde(\sqrt n)$ vertices and $\Oo(\sqrt n)$ edges and each vertex in $G_\mathrm{long}'$ can determine its neighbors in $\widetilde{\mathcal{O}}(\sqrt{ n })$ of \CONGEST. 
    Hence, we can broadcast $G_\mathrm{long}'$ to all nodes of $G$ in $\Otilde(\sqrt{ n } + D)$ rounds using \Cref{lm:aggregation}.
    Each vertex on a long path can now find the closest ruling set vertex in $R$ and use the graph $G_\mathrm{long}'$ to determine both endpoints of the path in $\widetilde{\mathcal{O}}(\sqrt{ n })$ rounds of \CONGEST.
    To sum up, it takes $\widetilde{\mathcal{O}}(\sqrt{ n } + D)$ rounds of \CONGEST on $G$ to compute the graph $G'' = \textit{Compress}(G')$ and to communicate the local topology of $G''$ to each vertex $v$ in $G$.
    Now that every vertex on a long path knows both its endpoints, any further round of communication on $G$ can be routed in $\Oo(\sqrt n+D)$ rounds, again using \Cref{lm:aggregation}.
\end{proof}

\RakeCompress*

\begin{proof}
    First we show that one iteration of \textit{Rake} and \textit{Compress} can be performed in $\widetilde{\mathcal{O}}(\sqrt{ n } + D)$ rounds of \CONGEST.
    The iterated \textit{Rake} and \textit{Compress} procedure produces a sequence of graphs starting from $H_0 := H$ until it terminates with a graph $H_k$ where each component is either a single cycle or has minimum degree of at least three.
    Note that each $H_i$ is a topological minor of the \CONGEST network $G$.
    Secondly, to complete the proof, we show that $k = \mathcal{O}(\log n)$.

    We maintain the invariant that a single round of \CONGEST on $H_i$ can be simulated in $\widetilde{\mathcal{O}}(\sqrt{ n } + D)$ rounds on $G$.
    Clearly, the invariant holds for $H_0 := H$ and due to the implementation of \textit{Compress} shown in \Cref{lem:compress} we can guarantee that it holds for any subsequent $H_i$.
    Performing one \textit{Rake} operation takes just a single round on $H_i$ and can therefore be simulated in $\widetilde{\mathcal{O}}(\sqrt{ n } + D)$ rounds of $G$.
    Further, each \textit{Compress} operation also takes $\widetilde{\mathcal{O}}(\sqrt{ n } + D)$ rounds on $G$, again according to \Cref{lem:compress}.

    Now we claim that this procedure terminates after $k = \lceil \log n \rceil $ iterations with a topological minor $H_{k}$ of $G$ without vertices of degree less than three.
    We show that after the first \textit{Rake} and \textit{Compress} iteration, every subsequent \textit{Rake} and \textit{Compress} iteration reduces the number of degree-one vertices by at least a factor two.
    Note that each \textit{Compress} iteration removes all vertices that currently have degree two from the graph. Furthermore, since we replace each path of degree-two vertices by an edge connecting its two endpoints, every remaining vertex maintains its degree.
    Hence, from the second \textit{Rake} operation onwards, there are no degree-two vertices in the graph before the \textit{Rake} operation.
    Therefore, every degree-one vertex $v$ is either adjacent to another degree-one vertex or a vertex of degree at least three.
    In the first case, \textit{Rake} removes both vertices.
    In the latter case, let $d \geq 3$ denote the degree of $v$'s only neighbor $w$.
    If $w$ has less than $d - 1$ neighbors of degree one, then its degree remains at least two after the \textit{Rake} operation.
    Hence, in order to create a new vertex of degree one, at least $d - 1 \geq 2$ vertices of degree one need to be removed from the graph.
    Let $\ell$ denote the number of degree-one vertices in $H$.
    Thus, after $k = \lceil \log \ell \rceil + 1 \leq \lceil \log n \rceil + 1$ iterations of \textit{Rake} and \textit{Compress}, the resulting graph $H_k$ does not contain any vertices of degree less than two.
    At this point, one additional \textit{Compress} operation removes all remaining degree-two vertices from the graph without decreasing the degree of any other vertex.
    The only case in which \textit{Compress} cannot remove all vertices of degree-two is if there is no vertex of degree three left.
    Hence, each component of the resulting graph is either a cycle or has minimum degree at least three.
\end{proof}

\paragraph{Rooting a forest.}

Next, we apply the \textit{Rake} and \textit{Compress} procedure to efficiently compute an orientation in an unrooted forest that for each tree $T$ selects a unique root $r$ and lets each vertex point towards its parent in $T$ with respect to $r$.

\begin{lemma}[Rooting all trees in a forest]
    \label{lm:root_a_tree}
    Let $G=(V,E)$ denote an $n$-node \CONGEST communication network of diameter $D$ and let $F \subseteq G$ be a forest. There is a deterministic $\widetilde{\mathcal{O}}(\sqrt{ n } + D)$-round \CONGEST algorithm that roots every tree $T\in F$ at an arbitrary vertex $r \in T$ and informs every node in $V$ of their parent in their respective tree.
\end{lemma}

\begin{proof}
    Let $T_0 := T$ and obtain $T_{i+1}$ from $T_i$ by applying one \textit{Rake} operation followed by one \textit{Compress} operation on $T_i$.
    From \Cref{lem:rake_and_compress} we get that after $k = \mathcal{O}(\log n)$ iterations, this sequence terminates with a tree $T_k$ of minimum degree at least three.
    Since the average degree of every tree is less than two, this implies that $T_k$ must be empty.




    Now we show how to use the decomposition to efficiently root $T$.
    We traverse the sequence of subtrees starting from the last non-empty graph $T_{k-1}$ to $T_0$.
    Note that the operation that removes the last vertex of the tree is always a \textit{Rake} operation and therefore $T_{k-1}$ is either just a single vertex or consists of two vertices connected by a single edge. In any case, we can pick a root $r$ in $T_{k-1}$ in a single round on $T_{k-1}$ and (if necessary) orient the only edge towards $r$.
    Now suppose all edges in $T_{i+1}$ are oriented towards a common root $r$.
    Note that any edge $e \in E(T_{i+1})$ corresponds to a path $P(e)$ in $T$ and according to \Cref{lem:compress} can broadcast a message to all edges $f \in P(e)$ in $\widetilde{\mathcal{O}}(\sqrt{ n } + D)$ rounds on $G$.
    For every edge $e = (u \to v) \in E(T_{i+1})$ we orient all edges on $P(e)$ towards $v$ in $\widetilde{\mathcal{O}}(\sqrt{ n } + D)$ rounds on $G$. This orients all edges that were removed from $T_i$ by the \textit{Compress} operation.
    For any vertex $v$ removed by the \textit{Rake} operation from $T_i$, we orient all incident edges $\{ u, v \} $ from $v$ to $u$. This operation takes a single round on $T_i$ and can thus be simulated in $\widetilde{\mathcal{O}}(\sqrt{ n } + D)$ rounds on $T$.
    Hence, after processing all subtrees $T_k,\dots,T_0 := T$ we have successfully rooted $T$.
    Since $k = \mathcal{O}(\log n)$, the whole algorithm runs in $\widetilde{\mathcal{O}}(\sqrt{ n } + D)$ rounds of \CONGEST.
\end{proof}

\section{Deferred Proofs}\label{sec:deferred_proofs}

In this section, we provide the missing proofs. First, we prove the relation between the two different formulations \ref{def:min-MPC} and \ref{def:max-MPC} of mixed packing-covering LPs.

\equivalentFormulations*

\begin{proof}
Let $(\x,\lambda)$ be a solution of the LP \ref{def:min-MPC}. Let $\widetilde{\x}=\x/\lambda$ and $\gamma = 1/\lambda$, then 
\begin{align*}
    \Pm \widetilde{\x}&=\Pm(\x/\lambda)\leq \lambda \pv / \lambda = \pv \\
    \Cm \widetilde{\x} & = \Cm(\x/\lambda)\geq \cv/\lambda = \gamma \cv.
\end{align*}
So $(\widetilde{\x}, \gamma)$ is a feasible solution of the LP \ref{def:max-MPC}.
Moreover, if $\lambda^* \leq \lambda \leq (1+\eps) \lambda^*$, we have
\begin{align*}
    1/\gamma^*=\lambda^*\leq \lambda  \leq (1+\eps) \lambda^* = (1+\eps)/\gamma^*.
\end{align*}
Since $ \gamma=1/\lambda$ we get $\gamma^*/(1+\eps) \leq \gamma \leq \gamma^*$.
\end{proof}

Secondly, we show that we can simulate the mixed packing-covering solver on the communication network given by the child gift graph.

\EquivalentModels*

\begin{proof}
    We may assume that every entity holds its corresponding constraints of \LP, i.e., each child has access to its corresponding constraints \Ref{eq:lp-objective-value}, \ref{eq:lp-gifts-sum-up-to-one}, and \ref{eq:lp-child-gets-only-one-big-gift}, each small gift has access to \ref{eq:lp-gifts-sum-up-to-one} and \ref{eq:lp-small-gift-assigned-to-beta-many-children}, and each big gift has access to \ref{eq:lp-big-gift-assigned-to-one-child}. Variables are defined on the edges between a child and a gift, so we may assume that children take charge of the variable nodes in \Cref{alg:mixed_covering/packing}. Packing/covering constraints are simulated by the corresponding entities that have access to the corresponding constraint. Finally, we observe that by the definition of \LP variables $x_{cg}$ for a child $c$ and a gift $g$ only show up if there exists an edge $\{c,g\}$ in $G$, i.e., entries in the packing/covering matrices are positive if and only if there exists an edge in $G$.
\end{proof}

Finally, we prove that we can implement~\Cref{alg:mixed_covering/packing} in the \CONGEST model. 

\finiteprecision*

\begin{proof}
    Let $\delta_{\mathrm{in}}:=\eps/100$. For every non-zero entry $\alpha$ of $\Pm$ and $\Cm$, replace it by a rational number in the interval $[(1-\delta_{\mathrm{in}})\alpha,(1+\delta_{\mathrm{in}})\alpha]$.
In particular, we obtain rational matrices $\widehat{\Pm},\widehat{\Cm}$ such that entrywise
\begin{align*}
(1-\delta_{\mathrm{in}})\Pm \leq \widehat{\Pm} \leq (1+\delta_{\mathrm{in}})\Pm
\quad \text{and} \quad
(1-\delta_{\mathrm{in}})\Cm \leq \widehat{\Cm} \leq (1+\delta_{\mathrm{in}})\Cm.
\end{align*}
Since $\Pm,\Cm$ have non-negative entries and $\x\ge 0$, for every $\x\ge 0$ we have 
\begin{align*}
(1-\delta_{\mathrm{in}})\Pm\x \leq \widehat{\Pm}\x \leq (1+\delta_{\mathrm{in}})\Pm\x
\quad \text{and} \quad
(1-\delta_{\mathrm{in}})\Cm\x \leq \widehat{\Cm}\x \leq (1+\delta_{\mathrm{in}})\Cm\x.
\end{align*}

\begin{claim}\label{claim:initial_rational}
    Any vector $\x \geq 0$ satisfying
    $\max\limits_j (\widehat{\Pm}\x)_j \leq (1+\eps/100)(1+\eps/3)\min\limits_j (\widehat{\Cm}\x)_j$
    also satisfies $\max\limits_j (\Pm\x)_j \leq (1+\eps)\min\limits_j (\Cm\x)_j$.
\end{claim}

\begin{proof}[Proof of \Cref{claim:initial_rational}]
    If $\x \geq 0$ satisfies
    $
    \max\limits_j (\widehat{\Pm}\x)_j \leq (1+\eps/100)(1+\eps/3)\min\limits_j (\widehat{\Cm}\x)_j,
    $
    then
    \begin{align*}
    \max_j (\Pm\x)_j
    \leq \frac{1}{1-\delta_{\mathrm{in}}}\max_j (\widehat{\Pm}\x)_j 
    &\leq \frac{(1+\eps/100)(1+\eps/3)}{1-\delta_{\mathrm{in}}}\min_j (\widehat{\Cm}\x)_j \\
    &\leq \frac{(1+\eps/100)(1+\eps/3)(1+\delta_{\mathrm{in}})}{1-\delta_{\mathrm{in}}}\min_j (\Cm\x)_j.
    \end{align*}
    For $\delta_{\mathrm{in}}=\eps/100$, the prefactor is upper bounded by
    $
    \frac{(1+\eps/100)^2(1+\eps/3)}{1-\eps/100}.
    $
    
    Using $(1+\eps/100)^2\le 1+3\eps/100$ and
    $1/(1-\eps/100)\le 1+2\eps/100$ for $\eps\in(0,1]$,
    we obtain
    \[
    \frac{(1+\eps/100)^2(1+\eps/3)}{1-\eps/100}
    \le
    (1+3\eps/100)(1+\eps/3)(1+2\eps/100)
    \le
    1+\eps.\qedhere
    \]
\end{proof}

Thus, solving the rounded instance to accuracy $(1+\eps/100)(1+\eps/3)$ is enough to obtain a $(1+\eps)$-approximate solution for the original instance.

\medskip
 
Next, we show that the rounded instance can be executed with fixed floating-point precision. 

\begin{claim}
    \label{claim:rational}
    There exists an absolute constant $c>0$ such that if every communicated number is rounded to relative precision
    $
    1\pm \delta
    \text{ with }
    \delta := \frac{\eps}{200cR},
    $
    where $R$ is the iteration bound from \Cref{thm:solve_feasibility}, then the
    finite-precision execution outputs a vector $\x$ satisfying
    \[
    \max_j (\widehat{\Pm}\x)_j
    \le
    (1+\eps/100)(1+\eps/3)\min_j (\widehat{\Cm}\x)_j.
    \]

    Additionally, every communicated number can be represented by $O(\log n)$ bits.
\end{claim}
If \Cref{claim:rational} holds, the final output is still a $(1+\eps)$-approximate feasible solution for the original instance.

\begin{proof}[Proof of \Cref{claim:rational}]
    We consider the execution of \Cref{alg:mixed_covering/packing} on the rounded rational instance $(\widehat{\Pm},\widehat{\Cm})$. Every quantity that is ever communicated is replaced by a rational approximation with relative error at most~$\delta=\eps/(200cR)$, where $c>0$ will be defined later in the proof:
    \begin{align*}
    \widetilde{u} \leftarrow (1+\delta)^{\lfloor \log_{1+\delta} u \rfloor}
    \end{align*}
    Here $\widetilde{u}$ denotes the rounded representation of any communicated quantity $u$.

    In one iteration, each communicated value is obtained from previously stored values by a constant number of additions, multiplications, divisions, and one evaluation of $\exp(\cdot)$. Since all quantities that occur before termination satisfy
    \begin{align*}
        0\leq (\widehat\Pm\x)_j,(\widehat\Cm\x)_j \leq K=\frac{10\ln n}{\eps},
    \end{align*}
    the arguments of all exponentials lie in $[-K,K]$, and hence all values $y_j=\exp((\widehat\Pm\x)_j)$ and $z_j=\exp(-(\widehat\Cm\x)_j)$ lie in a range whose logarithm is $O(K)$. Further, taking logarithms of the update rule
    \begin{align*}
        \x_i^{(t+1)}
        =
        \x_i^{(t)}
    \Bigl(1+\frac{\Delta_i^{(t)}}{K}\Bigr)
    \end{align*}
    yields
    \begin{align*}
        \log \x_i^{(t+1)} = \log \x_i^{(t)} + \log\Bigl(1+\frac{\Delta_i^{(t)}}{K}\Bigr).
    \end{align*}
    Since $0\le \Delta_i^{(t)}\le 1/2$ and
    $\log(1+u)\le u$ for all $u\ge0$, we obtain
    \begin{align*}
        \log\Bigl(1+\frac{\Delta_i^{(t)}}{K}\Bigr) \leq \frac{1}{2K}.
    \end{align*}
    Summing over all $R$ iterations shows that
    $\log \x_i^{(t)}$ increases by at most
    $O(R/K)$ over the entire execution. Consequently, every transmitted number $k \in \{y_j,z_j,x_i,a_{ij},b_{ij}\}$ can be stored in normalized floating-point form
    $
        k=\pm\, M\cdot 2^E
   $
    with $O(\log(R/\eps))$ bits for the mantissa $M$ and $O(\log K+\log R)$ bits for the exponent $E$, since $R=\poly(\log n,1/\eps,\log m)$ and $\eps\ge 1/\poly(n)$, this is $O(\log n)$ bits in total.

    It remains to bound the effect of arithmetic rounding. 
    Because all quantities occurring in the execution are positive and each newly stored value is computed from previously stored values by only constantly many operations, one iteration increases the current relative error by at most a multiplicative factor $1+c\delta$ for some absolute constant $c>0$. Hence, if after iteration $t$ every stored value approximates the exact value within a factor $1+\eta_t$, then after the next iteration the error is at most $        1+\eta_{t+1}\le (1+c\delta)(1+\eta_t).
    $
        Starting with $\eta_0=0$, induction yields
     $   1+\eta_t\le (1+c\delta)^t
        \le \exp(c t\delta).
    $
    Therefore, for all $t\le R$,
    $
        1+\eta_t\le \exp(cR\delta).
    $
    Choosing $\delta\le \eps/(200cR)$ implies
    \[
        \exp(cR\delta)\le \exp(\eps/200)\le 1+\eps/100,
    \]
    so the total arithmetic error over the entire execution is at most
    $\eps/100$ multiplicatively.
    \end{proof}
    Hence, the finite-precision execution computes a solution that is $(1+\eps/3)$-feasible for the rounded instance up to an additional factor $1+\eps/100$, and therefore, by \Cref{claim:initial_rational} it is still a $(1+\eps)$-approximate solution for the original instance.
    Finally, every communicated quantity is rational by construction, and since the update rule multiplies rational numbers by rational approximations, every output value $\x_i$ is rational as well.
\end{proof}

\end{document}